\documentclass[11pt]{article}

\usepackage[a4paper]{geometry}
\usepackage{amsmath,amssymb,amsthm}
\usepackage{mathtools}
\usepackage{bm,bbm}
\usepackage{graphicx}
\usepackage{hyperref}
\usepackage[authoryear,round]{natbib}
\usepackage{xcolor}
\usepackage{marginnote}
\usepackage{enumitem}
\usepackage{etoolbox}

\hypersetup{
  colorlinks=true,
  linkcolor=blue!50!black,
  citecolor=blue!50!black,
  urlcolor=blue!50!black
}

\theoremstyle{plain}
\theoremstyle{definition}

  \newtheorem{claimx}{Claim}

  \newtheorem{definition}{Definition}

  \newtheorem{lemma}{Lemma}
  \newtheorem{observation}{Observation}
  \newtheorem{proposition}{Proposition}
  \newtheorem{remark}{Remark}
  \theoremstyle{remark}

\newcommand{\R}{\mathbb{R}}
\newcommand{\X}{X}
\newcommand{\thH}{\theta_H}
\newcommand{\thL}{\theta_L}
\newcommand{\CFBH}{C^{FB}_H}
\newcommand{\CFBL}{C^{FB}_L}
\newcommand{\CaL}{C^{\alpha}_L}
\newcommand{\CbL}{C^{\beta}_L}
\newcommand{\muA}{\mu_{\alpha}}
\newcommand{\muB}{\mu_{\beta}}
\newcommand{\dirac}{\delta} 
\DeclareMathOperator*{\argmax}{arg\,max}

\title{Contracting for Information:\\ Heterogeneous Costs and Investment Opportunities}
\author{Han Wang\thanks{\href{mailto:han.wang1@ucr.edu}{han.wang1@ucr.edu}; Department of Economics, University of California, Riverside\\I am grateful to Yaron Azrieli, Paul J. Healy, James Peck, Ellen Muir, Teck Yong Tan, Mark Whitmeyer, Jo\~{a}o Thereze, Xiaoyu Cheng, Siyang Xiong, Kun Zhang and participants at the OSU Theory/Experimental Reading Group, the 2025 UNC Theory Workshop, the 2026 Southwest Economic Conference, the Spring 2026 Midwest International Trade \& Theory Conference for their helpful comments and suggestions.}}

\date{\today}

\begin{document}
\maketitle

\begin{abstract}
A principal faces a decision problem under uncertainty and can contract with a researcher to provide relevant information. The cost of acquiring information is only known to the researcher, and, moreover, by privately making an investment the researcher can reduce their expected cost. We characterize optimal contracts by reducing the problem to information design with cost constraints. The principal induces investment below a cutoff in its sunk cost. With two cost types, investment may be underprovided but is never excessive relative to the first best. With more types, the cutoff property survives but overinvestment can occur. Investment incentives also reshape information acquisition: optimal experiments may need to generate more distinct beliefs than there are possible states, and their design can be sensitive to the principal’s initial beliefs.
\end{abstract}

\newpage
\section{Introduction}\label{sec:intro}

Governments, firms, and other organizations rely on researchers to generate information for decisions under uncertainty. The cost of acquiring this information depends both on the researcher's private circumstances and on investments made before those circumstances are known. For example, researchers may develop methods, acquire equipment, or build data infrastructure before learning how costly a particular study will be. Such investments can lower subsequent research costs, but may be difficult for the principal to observe. This creates an interaction between moral hazard and adverse selection: the researcher privately chooses an investment that shapes future costs and subsequently learns the realized cost.

This paper studies contracting for information when hidden investment shapes the researcher's cost distribution. A principal faces a decision problem under uncertainty and hires a researcher to acquire relevant information. The researcher can design experiments that determine both how much to learn and what to learn about. The principal first commits to a \emph{methods-based contract}: a menu of experiments and associated payments. Payments can depend on the experiment but not on investment. The researcher then privately decides whether to invest before learning his cost type. Investment increases his probability of being efficient. The principal designs the menu to screen the researcher's cost type and provide investment incentives. Her objective is to maximize her expected decision payoff net of payments.

We begin with two cost types and characterize the optimal contract by comparing the principal's payoffs from inducing and deterring investment. Each program reduces to an information-design problem, possibly subject to a cost constraint, allowing us to draw on tools from that literature. The principal induces investment when doing so yields the higher payoff.

The optimal contract retains familiar screening features: the efficient type receives first-best information, while the inefficient type acquires no more than the first-best amount (Proposition~\ref{prop:dist}). Investment incentives can, however, eliminate the inefficient type's distortion, so both types receive first-best information for some parameter values. Investment is induced when its sunk cost is at or below a cutoff (Proposition~\ref{prop:cutoff}). This cutoff never exceeds the first-best threshold and may be strictly lower (Proposition~\ref{prop:under}). Underinvestment arises because the principal does not capture the full surplus gain from investment, even under full commitment.

Investment incentives also affect the structure of information acquisition. First, they can make more posteriors than states necessary in every optimal low-type experiment (Section~\ref{sec:info}). Second, they can alter how optimal support responds to the prior. At a fixed effective multiplier, an unconstrained optimum can retain its support as the prior varies within an affine segment of the concave envelope. Under a fixed positive cost cap, however, no support pair can remain optimal throughout an open interval of priors on which the cap binds and the optimum remains two-point (Section~\ref{sec:compstat}).

With finitely many cost types (Section~\ref{sec:ntype}), the principal's screening value remains non-decreasing under FOSD shifts toward more efficient types, and investment retains a cutoff structure. Overinvestment, however, becomes possible. In a three-type example, the principal strictly prefers to induce investment even when its cost exceeds the first-best gain in expected surplus (Proposition~\ref{prop:overinvestment}). The two-type assumption therefore has a substantive economic implication: it rules out excessive investment.

\paragraph{Related literature.} The paper is closest to the literature on contracting for information acquisition: \citet{osband1989optimal}, \citet{zermeno2011principal}, \citet{rappoport2017incentivizing}, \citet{carroll2019robust}, \citet{azrieli2021monitoring}, \citet{clark2021contracts}, \citet{azrieli2022delegated}, \citet{hafner2022young}, \citet{whitmeyer2022buying}, \citet{yoder2022designing}, \citet{li2023incentivizing}, \citet{min2023screening}, \citet{sharma2025procuring}, and \citet{wittbrodt2025delegating}. We assume throughout that the experiment itself is contractible, with payments conditioned on the experiment the researcher runs. Because this is favorable to the principal, our characterization can be read as an upper bound on what she can achieve under weaker contracting technologies: payments contingent on the realized state \citep{whitmeyer2022buying,li2023incentivizing,sharma2025procuring}, on the experiment's outcome or realized posterior \citep{rappoport2017incentivizing,yoder2022designing,wang2025contracting}, or on the cross-reports of several researchers through peer monitoring \citep{azrieli2021monitoring,azrieli2022delegated}. A recurring concern there is that the researcher may overstate his cost to earn higher pay. \citet{yoder2022designing} and \citet{wang2025contracting} study the screening problem when this cost is private. We extend this line by \emph{endogenizing} the cost type: the researcher makes a hidden, costly, upfront investment that shifts his type toward the efficient one.

This situates the paper within the literature on screening with endogenous types. \citet{gonzalez2004investment} studies unobserved cost-reducing investment undertaken before contracting. \citet{liu2022sequential} study sequential-screening procurement: an agent with initial private information invests under a committed contract, improving the distribution of a subsequently realized production cost. \citet{gershkov2021theory} likewise place commitment before investment, with privately informed bidders choosing actions that affect their valuations. In our model, the researcher has no private cost information before investment: the principal commits to a contract, investment shifts the distribution of the researcher's cost type, and that type is then realized. Our focus is on how these investment incentives interact with contracting over experiments. The resulting restrictions on information acquisition connect the screening problem to constrained information design \citep{letreust2019persuasion,azrieli2021constrained,doval2024constrained}.

Because the principal commits to the contract before investment is sunk, the underinvestment we find reflects a rent-sharing wedge rather than the ex-post hold-up familiar from the incomplete-contracts literature \citep{grossman1986costs,tirole1986procurement,hart1990property,segal2016property}. More recently, \citet{nguyen2019information} study \emph{information control} in the hold-up problem, where the strategic lever is the design of what a party learns. In our setting, information is instead the object the principal procures, and the resulting wedge persists under full commitment. \citet{dworczak2024mechanism} provide a mechanism-design treatment of property rights, where the allocation of control determines who is exposed to hold-up. Our model holds the allocation of control fixed and studies how screening rents and the allocation of experiments jointly shape investment incentives under commitment.

Our comparative statics with respect to the type distribution connect to \citet{maskin1984monopoly} and also relate to work on how the principal's information about the agent shapes screening. \citet{asseyer2025information} ranks signals that garble a fixed prior into posteriors, using the Blackwell order and, under a regularity condition, a hazard-rate-spread order, and traces their often non-monotone welfare effects. In contrast, we study a \emph{directed} shift from one prior to another and introduce an endogenous investment margin. Finally, our prior-sensitivity result connects to the rational-inattention notion of locally invariant posteriors \citep{caplin2022rationally}.

Section~\ref{sec:model} presents the model, and Section~\ref{sec:analysis} characterizes the optimal contract and its properties. Section~\ref{sec:example} uses a worked example to examine lumpy information acquisition, multi-posterior experiments, and prior sensitivity. Section~\ref{sec:ntype} extends the analysis to finitely many cost types, and Section~\ref{sec:conclude} concludes. The appendix contains the remaining proofs.

\section{Model Setting}\label{sec:model}

\paragraph{Overview.} We consider a model with two players: a principal (she) and a researcher (he). The principal is faced with a Bayesian decision problem and can hire the researcher to learn about an unknown state of the world. Suppose that the state of the world can take values in a finite set $\Omega$, with generic element $\omega$. The decision problem is a triple $(A,u,p_0)$, including a finite set of actions $A$, a utility function $u:A\times\Omega\to\R$, and a prior belief over the state $p_0\in\Delta(\Omega)$. We let $v(p):=\max_{a\in A}\sum_{\omega\in\Omega}p(\omega)\,u(a,\omega)$ be the highest expected utility attainable at belief $p\in\Delta(\Omega)$. As a maximum of finitely many affine functions, $v$ is convex and piecewise-linear.

The researcher shares the prior belief $p_0$ and can choose any experiment, represented by the distribution of posterior beliefs it induces. The feasible set is
\[
  \X:=\left\{
    \tau\in\Delta(\Delta(\Omega)):
    \int_{\Delta(\Omega)}p\,d\tau(p)=p_0
  \right\},
\]
where the mean restriction is Bayes plausibility.

The researcher's cost depends on a private type $\theta\in\Theta=\{\thH,\thL\}$ with $\thL>\thH>0$ ($\thH$ is the \emph{more efficient} type). For a type-$\theta$ researcher choosing $\tau$, the cost is $\theta\,C(\tau)$, where
\[
  C(\tau)=\mathbb{E}_{p\sim\tau}[c(p)]
\]
for a continuous and strictly convex $c:\Delta(\Omega)\to\R$ with $c(p_0)=0$.\footnote{It is immediate that $C(\tau)\ge 0$ for all $\tau\in\X$, with equality if and only if $\tau$ is the fully uninformative experiment (mass $1$ on $p_0$): by Jensen and $\mathbb{E}_{p\sim\tau}[p]=p_0$, $C(\tau)=\mathbb{E}[c(p)]\ge c(p_0)=0$, strict unless $\tau=\dirac_{p_0}$.} We write $\bar C:=\max_{\tau\in X}C(\tau)<\infty$. The type distribution is influenced by the researcher's investment.

\paragraph{Timeline.} At $t=0$ the principal commits to a methods-based contract. At $t=1$ the researcher decides whether to make a private, costly investment, which determines his type distribution. At $t=2$ the type realizes, the researcher chooses an experiment and is paid per the $t=0$ terms.

The principal has commitment power and posts the contract before the researcher prepares. The researcher then makes a project-specific, cost-reducing investment, such as tools, training, or a research pipeline, in anticipation of its terms. Two features of this timing are important. First, because investment precedes the realization of the type, it shifts the \emph{distribution} of the researcher's cost rather than responding to an already realized type. Second, because the principal contracts before the researcher invests, she shapes his investment incentives through the contract itself, as captured by the investment constraint \eqref{eq:ICI} below.

\paragraph{Methods-based contract.} At $t=0$ the principal chooses $M=(\chi,T)$, where $\chi:\Theta\to\X$ is the experiment choice function and $T:\Theta\to\R_+$ the payment function. A type-$\theta$ researcher who chooses $\tau$ for payment $t$ gets utility $t-\theta C(\tau)$. Incentive compatibility and individual rationality require, for all $\theta,\theta'\in\Theta$,
\begin{align}
  T(\theta)-\theta C(\chi(\theta)) &\ge T(\theta')-\theta C(\chi(\theta')) \tag{IC}\label{eq:IC}\\
  T(\theta)-\theta C(\chi(\theta)) &\ge 0. \tag{IR}\label{eq:IR}
\end{align}

Write $U_\theta:=T(\theta)-\theta C(\chi(\theta))$ for the researcher's interim rent, with $U_H:=U_{\thH}$ and $U_L:=U_{\thL}$.

\paragraph{Investment.} At $t=1$ the researcher makes a binary investment decision at sunk cost $\phi>0$. Without investment, the probability of being the high type $\thH$ is $\alpha$; with investment, it is $\beta$, where $\beta>\alpha$. Given $M$ satisfying \eqref{eq:IC}--\eqref{eq:IR}, the researcher is willing to invest if and only if
\begin{equation}
  (\beta-\alpha)\Bigl[\bigl(T(\thH)-\thH C(\chi(\thH))\bigr)-\bigl(T(\thL)-\thL C(\chi(\thL))\bigr)\Bigr]\ge\phi.
  \tag{IC-I}\label{eq:ICI}
\end{equation}

If the reverse (weak) inequality holds we call it (IC-NI).

\paragraph{Principal's problem.} The principal's value from an experiment $\tau$ is $V(\tau)=\mathbb{E}_{p\sim\tau}[v(p)]$, the expectation of $v$ over the induced posteriors. Define
\[
  \check V(\chi,T)=
  \begin{cases}
    \alpha[V(\chi(\thH))-T(\thH)]+(1-\alpha)[V(\chi(\thL))-T(\thL)] & \text{if }(\chi,T)\text{ satisfies (IC-NI)},\\[2pt]
    \beta[V(\chi(\thH))-T(\thH)]+(1-\beta)[V(\chi(\thL))-T(\thL)]   & \text{if }(\chi,T)\text{ satisfies \eqref{eq:ICI}}.
  \end{cases}
\]

At a contract where \eqref{eq:ICI} holds with equality, both branches apply and the researcher is indifferent. We then take $\check V$ to be the larger of the two, breaking the tie in the principal's favor. The principal solves $\max_{(\chi,T)}\check V(\chi,T)$ subject to \eqref{eq:IC}--\eqref{eq:IR}.

\paragraph{Primitives.} Fixing the state space $\Omega$ and type space $\Theta=\{\thH,\thL\}$, a contracting environment is a tuple $(A,u,p_0,c,\alpha,\beta,\phi)$: the decision problem $(A,u,p_0)$, the belief-cost $c$, the pre- and post-investment type distributions $\alpha,\beta$, and the sunk investment cost $\phi$. Throughout, the prior $p_0\in\Delta(\Omega)$ and the type distributions have full support: $p_0(\omega)>0$ for every $\omega\in\Omega$, and $\alpha,\beta\in(0,1)$ (both types occur before and after investment). Section~\ref{sec:ntype} takes the finite-$\Theta$ analog.

\section{Analysis}\label{sec:analysis}

\subsection{Two programs}
We split the problem into
\begin{align}
  V_{NI}=\max_{(\chi,T)}\ &\alpha[V(\chi(\thH))-T(\thH)]+(1-\alpha)[V(\chi(\thL))-T(\thL)]\notag\\
   &\text{s.t. \eqref{eq:IC}, \eqref{eq:IR}, (IC-NI)} \tag{P-NI}\label{eq:PNI}\\[4pt]
  V_{I}=\max_{(\chi,T)}\ &\beta[V(\chi(\thH))-T(\thH)]+(1-\beta)[V(\chi(\thL))-T(\thL)]\notag\\
   &\text{s.t. \eqref{eq:IC}, \eqref{eq:IR}, \eqref{eq:ICI}} \tag{P-I}\label{eq:PI}
\end{align}
and take the better of the two. When their values coincide, we select an investment-inducing contract. Note that Program \eqref{eq:PI} can be infeasible if the cost of investment becomes too large. In that case, we set $V_{I}=-\infty$. 

For what follows, we introduce two orders on experiments. For $\tau_1,\tau_2\in\X$, write $\tau_2\succeq_{c}\tau_1$ (the \emph{cost order}) if $C(\tau_2)\ge C(\tau_1)$, that is $\mathbb{E}_{p\sim\tau_2}[c(p)]\ge\mathbb{E}_{p\sim\tau_1}[c(p)]$; write $\tau_2\succeq_{B}\tau_1$ (the \emph{Blackwell order}) if $\tau_2$ dominates $\tau_1$ in the Blackwell order \citep{blackwell1951comparison,blackwell1953equivalent}. Since $c$ is convex, a Blackwell-more-informative experiment is more costly, so $\succeq_{B}$ refines $\succeq_{c}$.

\begin{lemma}[Cost monotonicity]\label{lem:mono}
Let $Y\subseteq\X$ be any feasible set. For $\theta_1>\theta_2$, let
$\tau_i\in\argmax_{\tau\in Y}\{V(\tau)-\theta_i C(\tau)\}$, $i=1,2$. Then
$C(\tau_1)\le C(\tau_2)$.
\end{lemma}
\begin{proof}
Optimality of $\tau_1$ at $\theta_1$: $V(\tau_1)-\theta_1 C(\tau_1)\ge V(\tau_2)-\theta_1 C(\tau_2)$.
Optimality of $\tau_2$ at $\theta_2$: $V(\tau_2)-\theta_2 C(\tau_2)\ge V(\tau_1)-\theta_2 C(\tau_1)$.
Adding and canceling $V(\tau_1)+V(\tau_2)$ gives
$(\theta_1-\theta_2)C(\tau_2)\ge(\theta_1-\theta_2)C(\tau_1)$, hence $C(\tau_1)\le C(\tau_2)$.
\end{proof}

Lemma~\ref{lem:mono} orders the optimal experiments by cost: it gives $\tau_2\succeq_{c}\tau_1$ for the lower-multiplier optimum $\tau_2$. Its proof is a ``revealed-preference'' argument that adds two optimality inequalities, so it does not rely on the posterior-separable cost structure and holds for any finite $\Omega$. The following strengthens this to the Blackwell order $\succeq_{B}$: a more efficient (lower-cost) type is assigned a Blackwell-more-informative experiment. In a binary-state model, \citet{yoder2022designing} establishes Blackwell monotonicity.  For general finite state spaces, \citet{wang2025contracting} shows that Blackwell monotonicity need not hold.

\begin{lemma}[Blackwell monotonicity; two states]
\label{lem:blackwell}
Let $|\Omega|=2$ and $\theta_1>\theta_2>0$. For any $\tau_i\in\argmax_{\tau\in\X}\{V(\tau)-\theta_i C(\tau)\}$, $i=1,2$, we have $\tau_2\succeq_B\tau_1$, with strict dominance whenever
$C(\tau_2)>C(\tau_1)$.
\end{lemma}
\begin{proof}
See Appendix~\ref{app:blackwell}.
\end{proof}

We first define the scaled investment cost and the reference information costs used below. Let
\[
  \eta:=\frac{\phi}{(\beta-\alpha)(\thL-\thH)}.
\]

We also include $\phi=\eta=0$ as a boundary case. The scaled cost $\eta$ is the information-cost threshold in the reduced programs below. We state investment comparisons in $\phi$ and use $\eta$ to describe experiment choices. When writing $V_I(\eta)$ or $V_{NI}(\eta)$, we evaluate the corresponding program at $\phi=(\beta-\alpha)(\thL-\thH)\eta$.

Define the low type's \emph{effective (rent-inflated)} cost multipliers
\[
  \mu_\alpha:=\thL+\frac{\alpha}{1-\alpha}(\thL-\thH),
  \qquad
  \mu_\beta:=\thL+\frac{\beta}{1-\beta}(\thL-\thH),
\]
and let $K(\lambda):=\max_{\tau\in\X}\{V(\tau)-\lambda C(\tau)\}$, also written as $K_\lambda$.

For the two-type analysis, we assume that $K$ is differentiable at $\thH,\thL,\muA$, and $\muB$. Every optimizer at one of these reference multipliers $\lambda$ therefore has the same information cost and value, given by $-K'(\lambda)$ and $K(\lambda)-\lambda K'(\lambda)$, respectively.

Define the corresponding reference costs by
\[
  \CFBH:=-K'(\thH),\quad
  \CFBL:=-K'(\thL),\quad
  \CaL:=-K'(\muA),\quad
  \CbL:=-K'(\muB).
\]

For general finite state spaces, common cost and value need not identify a unique experiment. Replacing an assigned reference optimizer by another at the same multiplier nevertheless preserves all contract constraints and the principal's payoff, with transfers unchanged. A fixed reference experiment can therefore be used whenever the assigned experiment is optimal at that multiplier. For binary states, Lemma~\ref{lem:binary-unique} in Appendix~\ref{app:capped} establishes uniqueness at the reference multipliers.

Since $\thH<\thL<\muA<\muB$, Lemma~\ref{lem:mono} gives $\CFBH\ge\CFBL\ge\CaL\ge\CbL\ge0$. We assume $\CaL>0$: the screening benchmark under the pre-investment distribution assigns positive information cost to the low type. We impose no further strict inequalities among the reference costs.

\begin{lemma}[Simplifying \eqref{eq:PNI}]\label{lem:PNI}
If $(\chi^*,T^*)$ solves \eqref{eq:PNI}, then
\begin{enumerate}[label=\arabic*.]
  \item $\chi^*$ satisfies
    \begin{enumerate}[label=(\alph*)]
      \item $\chi^*(\thH)\in\argmax_\tau\{V(\tau)-\thH C(\tau)\}$,
      \item $\chi^*(\thL)\in\argmax_\tau\{V(\tau)-\muA C(\tau)\}$ subject to $C(\tau)\le\eta$.
    \end{enumerate}
  \item $T^*(\thH)=\thH C(\chi^*(\thH))+(\thL-\thH)C(\chi^*(\thL))$ and $T^*(\thL)=\thL C(\chi^*(\thL))$.
\end{enumerate}
\end{lemma}
\begin{proof}
See Appendix~\ref{app:PNI}.
\end{proof}

\begin{lemma}[Simplifying \eqref{eq:PI}]\label{lem:PI}
Suppose $(\chi^*,T^*)$ solves \eqref{eq:PI}.

\smallskip\noindent\emph{If $\eta\ge\CbL$ (``Case~1''):}
\begin{enumerate}[label=(1\alph*),leftmargin=\dimexpr\leftmargini+\leftmarginii\relax]
  \item $\chi^*(\thH)\in\argmax_\tau\{V(\tau)-\thH C(\tau)\}$ subject to $C(\tau)\ge\eta$,
  \item $\chi^*(\thL)\in\argmax_\tau\{V(\tau)-\thL C(\tau)\}$ subject to $C(\tau)\le\eta$,
\end{enumerate}
with $T^*(\thH)=\thH C(\chi^*(\thH))+(\thL-\thH)\eta$ and $T^*(\thL)=\thL C(\chi^*(\thL))$.

\smallskip\noindent\emph{If $\eta<\CbL$ (``Case~2''):}
\begin{enumerate}[label=(2\alph*),leftmargin=\dimexpr\leftmargini+\leftmarginii\relax]
  \item $\chi^*(\thH)\in\argmax_\tau\{V(\tau)-\thH C(\tau)\}$ subject to $C(\tau)\ge\eta$,
  \item $\chi^*(\thL)\in\argmax_\tau\{V(\tau)-\muB C(\tau)\}$ subject to $C(\tau)\ge\eta$,
\end{enumerate}
with $T^*(\thH)=\thH C(\chi^*(\thH))+(\thL-\thH)C(\chi^*(\thL))$ and $T^*(\thL)=\thL C(\chi^*(\thL))$.
\end{lemma}
\begin{proof}
See Appendix~\ref{app:PI}.
\end{proof}

\begin{proposition}\label{prop:reduce}
In either feasible program \eqref{eq:PNI} or \eqref{eq:PI}, each type's assigned experiment solves an information-design problem, possibly subject to an information-cost constraint.
\end{proposition}
\begin{proof}
By Lemmas~\ref{lem:PNI}--\ref{lem:PI}, each assigned experiment maximizes $V(\tau)-\lambda C(\tau)=\mathbb E_\tau[v(p)-\lambda c(p)]$ over the Bayes-plausible set $\X$, possibly subject to $C(\tau)\le\eta$ or $C(\tau)\ge\eta$. This is an information-design problem with an adjusted value function and, where applicable, a linear information-cost constraint.
\end{proof}

This reduction connects the contracting problem to constrained information design \citep{letreust2019persuasion,azrieli2021constrained,doval2024constrained}.

\subsection{Properties of the optimal contract}
Let $\chi^{FB}$ denote the principal's first-best experiment choice function (Appendix~\ref{app:FB}). The next proposition helps clarify how the optimal contract in our setting compares with the first-best benchmark in terms of information acquisition.

\begin{proposition}[Screening distortions]
\label{prop:dist}
Let $(\chi^*,T^*)$ solve \eqref{eq:PNI} or a feasible \eqref{eq:PI}. In either program, $\chi^*(\thL)\preceq_c\chi^{FB}(\thL)$. For binary states, $\chi^*(\thL)\preceq_B\chi^{FB}(\thL)$. Moreover:
\begin{enumerate}[label=\arabic*.]
  \item \emph{Without investment.} Under \eqref{eq:PNI}, the high type's experiment is first-best optimal.

  \item \emph{With investment.} Under \eqref{eq:PI}, the high type's experiment is first-best optimal when $\eta\le\CFBH$; otherwise, $C(\chi^*(\thH))=\eta>\CFBH$. The low type's experiment is first-best optimal when $\eta\ge\CFBL$.
\end{enumerate}
\end{proposition}
\begin{proof}
See Appendix~\ref{app:props23}.
\end{proof}

\begin{proposition}[Investment cutoff]\label{prop:cutoff}
There exists a cutoff $\phi_I$ such that the principal prefers investment if and only if $\phi\le\phi_I$.
\end{proposition}
\begin{proof}
See Appendix~\ref{app:cutoff}.
\end{proof}

The principal's net gain from inducing investment, $D(\phi)=V_I(\phi)-V_{NI}(\phi)$, is non-increasing in $\phi$. A higher sunk cost increases the rent wedge required to induce investment, tightening (IC-I), while making investment easier to deter, slackening (IC-NI). Therefore, once inducing investment ceases to be worthwhile, it never becomes worthwhile again. The principal thus induces investment precisely when $\phi\le\phi_I$.

The principal never induces inefficient investment, but may forgo efficient investment: her cutoff $\phi_I$ can fall strictly below the first-best threshold. To state the result, first define efficient investment. Investment is \emph{efficient} if and only if $\phi\le\phi^{FB}$, where $\phi^{FB}:=(\beta-\alpha)\Bigl[\bigl(V(\chi^{FB}(\thH))-\thH C(\chi^{FB}(\thH))\bigr)-\bigl(V(\chi^{FB}(\thL))-\thL C(\chi^{FB}(\thL))\bigr)\Bigr]$. Recall that $K_\theta$ is the first-best surplus for type $\theta$. Because the low type acquires information at the first best ($\CFBL\ge\CaL>0$) and $\chi^{FB}(\thL)$ is feasible in the $\thH$-problem, $K_{\theta_H}\ge K_{\theta_L}+(\thL-\thH)\CFBL>K_{\theta_L}$. Therefore, the efficient investment cutoff is $\phi^{FB}=(\beta-\alpha)(K_{\theta_H}-K_{\theta_L})>0$.

\begin{proposition}[Underinvestment]\label{prop:under}
The principal never over-invests: $\phi_I\le\phi^{FB}$. Hence investment, when induced, is efficient; and on the interval $(\phi_I,\phi^{FB})$, nonempty in the Example (Section~\ref{sec:example}), investment is efficient yet not induced.
\end{proposition}
\begin{proof}
See Appendix~\ref{app:under}.
\end{proof}

Write $\eta_I:=\phi_I/[(\beta-\alpha)(\thL-\thH)]$ and $\eta_{FB}:=\phi^{FB}/[(\beta-\alpha)(\thL-\thH)]$ for the scaled investment cutoffs. Proposition~\ref{prop:under} gives $\eta_I\le\eta_{FB}$, while evaluating the high type's first-best experiment in the low type's objective gives $K_{\theta_L}\ge K_{\theta_H}-(\thL-\thH)\CFBH$, and hence $\eta_{FB}\le\CFBH$. Thus the high type receives first-best information whenever investment is induced: the upward distortion identified in Proposition~\ref{prop:dist} never occurs in an optimal contract.

The next lemma complements this upper bound with a lower bound, $\eta_I\ge\CaL$. It also shows that, at the cutoff, the principal's payoff from deterring investment equals $V_\alpha$, her optimal screening payoff under the pre-investment type distribution (Appendix~\ref{app:obs}). This equality will help characterize the underinvestment wedge.
\begin{lemma}[A lower bound on the investment cutoff]\label{lem:etaIlb} 
The cutoff satisfies $\phi_I\ge(\beta-\alpha)(\thL-\thH)\CaL$ (equivalently, in information units, $\eta_I\ge\CaL$). Hence $V_{NI}(\phi_I)=V_\alpha$.
\end{lemma}
\begin{proof}
See Appendix~\ref{app:etaIlb}.
\end{proof}

The lemma places the cutoff in the region where the unconstrained screening benchmark under $\alpha$ remains feasible for the deterrence program. At $\eta_I=\CaL$, the deterrence constraint binds; at $\eta_I>\CaL$, it is slack. In either case, deterring investment entails no additional loss relative to that benchmark: $V_{NI}(\phi_I)=V_\alpha$.

Together, Proposition~\ref{prop:under} and Lemma~\ref{lem:etaIlb} bracket the cutoff between the low type's $\alpha$-cost and the efficient level:
\[
  (\beta-\alpha)(\thL-\thH)\CaL\ \le\ \phi_I\ \le\ \phi^{FB}.
\]

The next result quantifies underinvestment. When the low type's first-best information is feasible at the cutoff, the wedge $\phi^{FB}-\phi_I$ has an exact closed form; otherwise, that expression is a strict lower bound.

\begin{proposition}[Exact underinvestment wedge]\label{prop:wedge}
If the primitive inequality
\begin{equation}
  \beta K_{\theta_H}+(1-\beta)K_{\theta_L}-\beta(\thL-\thH)\CFBL\ \ge\ V_\alpha \tag{R}\label{eq:regimeR}
\end{equation}
holds (equivalently, $\eta_I\ge\CFBL$), then $\phi^{FB}-\phi_I=\frac{\beta-\alpha}{\beta}\,(V_\alpha-K_{\theta_L})$. If \eqref{eq:regimeR} fails (equivalently, $\CaL\le\eta_I<\CFBL$), then $\phi^{FB}-\phi_I>\frac{\beta-\alpha}{\beta}\,(V_\alpha-K_{\theta_L})$.
\end{proposition}
\begin{proof}
See Appendix~\ref{app:wedge}.
\end{proof}

Under~\eqref{eq:regimeR}, the low type receives first-best information at the cutoff, and the underinvestment wedge equals $(\beta-\alpha)/\beta$ times the principal's payoff advantage from optimal screening over pooling, $V_\alpha-K_{\theta_L}$. This payoff difference is not itself an information rent: both contracts generally leave rent to the high type. The information cap binds without distortion at $\eta_I=\CFBL$ and is slack for $\eta_I>\CFBL$. When~\eqref{eq:regimeR} fails, the cap instead forces information cost below the first-best level, $C(\chi^*(\thL))=\eta_I<\CFBL$, making the underinvestment wedge strictly larger than the closed-form benchmark.

\begin{remark}[The payoff gain from screening]
\label{rem:wedge}
Let $W_L^\alpha:=V(\chi^\alpha(\thL))-\thL\CaL$ be the total surplus generated by the low type's assigned experiment. Then
\[
\begin{aligned}
  V_\alpha-K_{\theta_L}
  ={}&\alpha(K_{\theta_H}-K_{\theta_L})
  -(1-\alpha)(K_{\theta_L}-W_L^\alpha)
      -\alpha(\thL-\thH)\CaL.
\end{aligned}
\]

The first term reflects the difference between the two first-best surplus levels. The remaining terms subtract the low type's allocation loss and the expected screening rent. Together they give the principal's payoff advantage over pooling. The investment cutoff equals the efficient threshold precisely when this advantage vanishes, $V_\alpha=K_{\theta_L}$.
\end{remark}

The investment cutoff also responds monotonically to the effectiveness of investment and to the researcher's initial type distribution, as the next proposition shows.

\begin{proposition}[Comparative statics of the investment cutoff]\label{prop:cutoffcs}
The investment cutoff $\phi_I$ is \emph{strictly} increasing in the post-investment high-type probability $\beta$ and \emph{strictly} decreasing in the pre-investment prior $\alpha$.
\end{proposition}
\begin{proof}
See Appendix~\ref{app:cutoffcs}.
\end{proof}
The comparative statics follow from $\phi_I=(\beta-\alpha)(\thL-\thH)\eta_I$. Raising $\beta$ weakly increases $\eta_I$, while raising $\alpha$ weakly decreases it. Since $\eta_I\ge\CaL>0$, the corresponding strict change in $\beta-\alpha$ makes the comparative statics of $\phi_I$ strict.

Under~\eqref{eq:regimeR}, both types receive first-best information at the cutoff, and $V_I(\eta_I)=V_\alpha$ gives
\[
  \eta_I=
  \frac{\beta K_{\theta_H}+(1-\beta)K_{\theta_L}-V_\alpha}
       {\beta(\thL-\thH)}.
\]

Outside this regime, the proof establishes the same comparative statics by comparing the non-increasing function $V_I(\eta)$ with the screening value $V_\alpha$, without requiring a closed-form expression for $\eta_I$.

\section{An Example: Lumpy Information and Prior Sensitivity}\label{sec:example}

We now apply our general results to a specific contracting environment and fully characterize the optimal contract. Beyond illustrating how to solve for the optimum, the example will also reveal several distinctive features that are difficult to see from the general analysis alone.

Throughout this section the state is binary, $\Omega=\{\omega_0,\omega_1\}$. We identify a belief $p$ with the scalar $p(\omega_1)\in[0,1]$, so that $p_0\in[0,1]$ and $v,c$ are functions on $[0,1]$. Fix the action space $A=\{l,m,r\}$ and the type space $\Theta=\{\thL,\thH\}$, and let $\phi\ge0$ be free. The parameters are $p_0=\tfrac12$; $\thH=\tfrac12$, $\thL=1$; $\alpha=\tfrac14$, $\beta=\tfrac34$ (so $\muA=\tfrac76$, $\muB=\tfrac52$); payoffs $u(l,\cdot)=(1,-2)$, $u(m,\cdot)=(0,0)$, $u(r,\cdot)=(-2,1)$ (here and throughout, payoff vectors are ordered $(\omega_0,\omega_1)$); and mutual-information cost $c(p)=H(p_0)-H(p)$ with $H(p)=-\sum_{\omega}p(\omega)\log p(\omega)$. The decision problem and information cost are those of the rational-inattention example in \citet{azrieli2021constrained}, whose characterization of the underlying information-acquisition problem serves as a building block for our contracting analysis.

For interpretation, suppose the principal must choose among three policy alternatives --- two opposing policies ($l$ and $r$) and a safe option ($m$) --- whose relative merits depend on the efficacy of a new vaccine. She hires a researcher who can invest before learning his cost of studying the vaccine. We assume that the study protocol and its implementation are verifiable, so the experiment is contractible, whereas investment remains private.

The computed reference costs are $\CbL=0$, $\CaL\approx0.437$, $\CFBL\approx0.502$, $\CFBH\approx0.676$. The knife-edge multiplier $\theta^{*}$ is the cost multiplier at which the optimal experiment switches between the uninformative experiment and the extreme (widest) two-posterior spread. Formally, $\theta^{*}=\sup\Lambda$ as defined in Appendix~\ref{app:mes3}. The resulting constants are
\[
\begin{aligned}
  \theta^{*}&=\tfrac{1}{\ln((1+\sqrt5)/2)}, & p^{*}&=[1+\exp(3/\theta^{*})]^{-1}\approx0.191,\\
  \eta_C&=H(\tfrac12)-H(p^{*})\approx0.206, & \eta_I&=\tfrac{\phi_I}{(\beta-\alpha)(\thL-\thH)}\approx0.555,
\end{aligned}
\]
with $\eta_C\in(\CbL,\CaL)$ and $\eta_I\in(\CFBL,\CFBH)$.

\begin{claimx}[Optimal experiments]\label{claim:example}
For $\thH$: $\tau(l)=\tau(r)=\tfrac12$ and $p_l=1-p_r=[1+\exp(3/\thH)]^{-1}$. For $\thL$, depending on $\eta$:
\begin{itemize}
  \item $\eta\in[0,\eta_C)$: $\tau(l)=\tau(r)=\frac{\eta}{2[H(1/2)-H(p^{*})]}$, $\tau(m)=1-\frac{\eta}{H(1/2)-H(p^{*})}$, $p_l=1-p_r=p^{*}$, $p_m=\tfrac12$;
  \item $\eta\in[\eta_C,\CFBL)$: $\tau(l)=\tau(r)=\tfrac12$, posteriors $p_l=1-p_r$ set by $H(\tfrac12)-H(p_l)=\eta$;
  \item $\eta\in[\CFBL,\eta_I]$: $\tau(l)=\tau(r)=\tfrac12$, $p_l=1-p_r=[1+\exp(3/\thL)]^{-1}$;
  \item $\eta>\eta_I$: $\tau(l)=\tau(r)=\tfrac12$, $p_l=1-p_r=[1+\exp(3/\muA)]^{-1}$.
\end{itemize}
\end{claimx}
\begin{proof}
See Appendix~\ref{app:example}.
\end{proof}

In words, the high type always runs his first-best experiment, while the low type's experiment passes through four regimes as the budget $\eta$ grows: a three-posterior straddle $\{p^{*},\tfrac12,1-p^{*}\}$ for $0<\eta<\eta_C$; a symmetric two-posterior spread pinned by the budget, $H(\tfrac12)-H(p_l)=\eta$, for $\eta\in[\eta_C,\CFBL)$; the low type's own first-best two-posterior experiment for $\eta\in[\CFBL,\eta_I]$; and the rent-inflated $\muA$-experiment for $\eta>\eta_I$.

\begin{observation}[Information acquisition]\label{obs:info}
(1) For $0<\eta<\eta_C$, the optimal contract assigns $\thL$ an experiment inducing \emph{three} posteriors, and no optimal experiment uses fewer than three. By contrast, unconstrained information design admits an optimal experiment with at most as many posteriors as states. The additional posterior is consistent with the support bound for constrained information design: the number of states plus the number of constraints \citep{doval2024constrained}. (2) For $\eta\in[\CFBL,\eta_I]$, the optimal contract achieves first-best information for \emph{both} types. In Case~1 of Lemma~\ref{lem:PI}, inducing investment fixes the high type's rent at $(\thL-\thH)\eta$. Varying the low type's experiment within the cap $C(\chi(\thL))\le\eta$ therefore does not change this rent. The principal consequently chooses the low type's experiment using the physical cost parameter $\thL$, rather than the rent-inflated multiplier $\muA$, subject to the cap. The cap binds without distortion at $\eta=\CFBL$ and is slack for $\eta>\CFBL$. A larger sunk cost loosens this cap, progressively reducing the low type's distortion until investment ceases beyond $\eta_I$.
\end{observation}

\begin{observation}[Investment]\label{obs:invest}
(1) The optimal contract induces investment if and only if $\eta\le\eta_I$. (2) $\eta_I<\eta_{FB}$, where $\eta_{FB}=\frac{\phi^{FB}}{(\beta-\alpha)(\thL-\thH)}\approx0.598$. Hence there is a nonempty interval $(\eta_I,\eta_{FB})$ on which efficient investment is \emph{not} induced, so the inequality $\phi_I\le\phi^{FB}$ of Proposition~\ref{prop:under} is strict here.
\end{observation}

Figure~\ref{fig:invest} illustrates these observations. It plots the principal's expected payoff under investment, $V_I(\eta)$, and under no investment, $V_{NI}(\eta)$. The two curves intersect at the cutoff $\eta_I$, and the shaded interval $(\eta_I,\eta_{FB})$ is the region in which efficient investment is not induced. Figure~\ref{fig:welfare} decomposes the resulting surplus, net of the investment cost $\phi$. While investment is induced, the principal's payoff falls as the required rent rises. Above $\eta_I$, her payoff is constant in this example. The gap between total welfare and her payoff is the researcher's expected rent net of any investment cost incurred. As $\eta$ rises above $\eta_I$, investment ceases, and total welfare and the researcher's payoff drop.

\begin{figure}[tbp]
\centering
\includegraphics[width=0.6\textwidth]{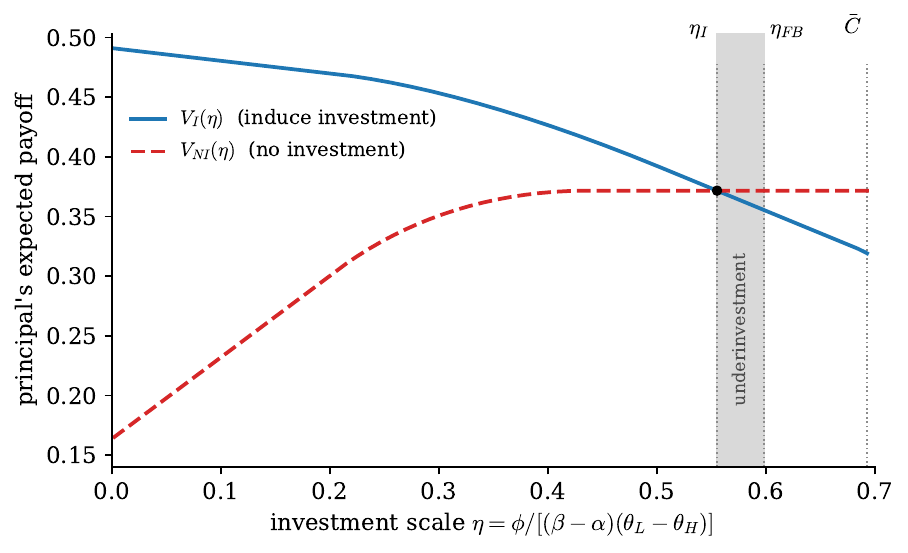}
\caption{The principal's expected payoffs under investment, $V_I(\eta)$ (solid), and under no investment, $V_{NI}(\eta)$ (dashed), as functions of the investment scale $\eta$. Investment is induced where $V_I\ge V_{NI}$, that is $\eta\le\eta_I$. The curves cross at $\eta_I\approx0.555$. The shaded interval $(\eta_I,\eta_{FB})$, with $\eta_{FB}\approx0.598$, is where investment is efficient yet not induced (Proposition~\ref{prop:under}).}
\label{fig:invest}
\end{figure}

\begin{figure}[tbp]
\centering
\includegraphics[width=0.6\textwidth]{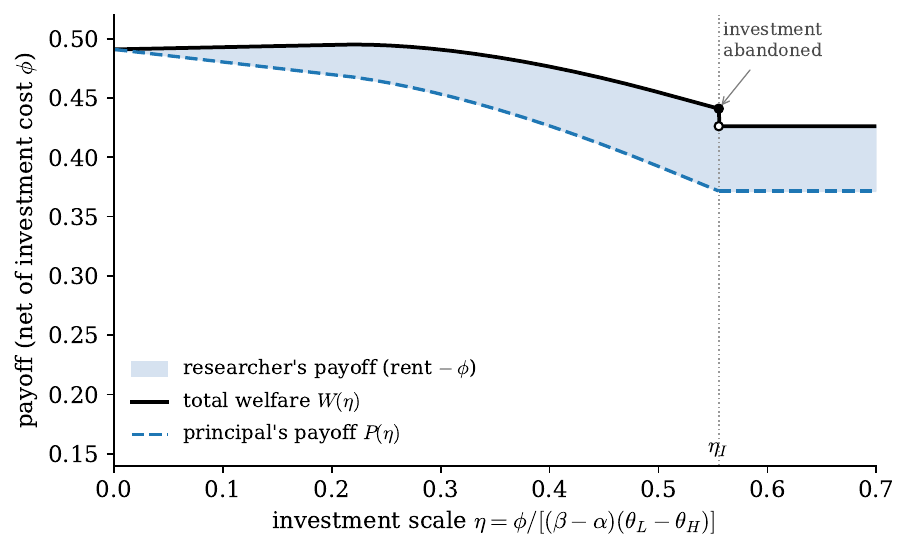}
\caption{Welfare decomposition. The solid line is total welfare $W(\eta)$, the dashed line the principal's payoff
$P(\eta)=\max\{V_I,V_{NI}\}$, and the shaded gap the researcher's expected rent net of any investment cost incurred. As $\eta$ rises above $\eta=\eta_I$, the selected contract switches to deterrence: total welfare and the researcher's payoff drop discontinuously (closed and open circles mark the left and right limits), while the principal's payoff is continuous. 
}
\label{fig:welfare}
\end{figure}

\subsection{Minimum efficient scale and multi-posterior experiments}\label{sec:info}
The three-posterior feature of Observation~\ref{obs:info}(1) is not incidental. It is driven by a gap in the costs of the optimal experiments. We derive this gap and then use it to characterize exactly when the multi-posterior regime arises. First, a \emph{minimum efficient scale} makes information \emph{lumpy}. This lumpiness then creates a cost gap in the reference costs that anchor the optimal contract. The key is which side of this gap contains the low type's post-investment screening benchmark cost $\CbL$.

\begin{figure}[tbp]
\centering
\includegraphics[width=0.8\textwidth]{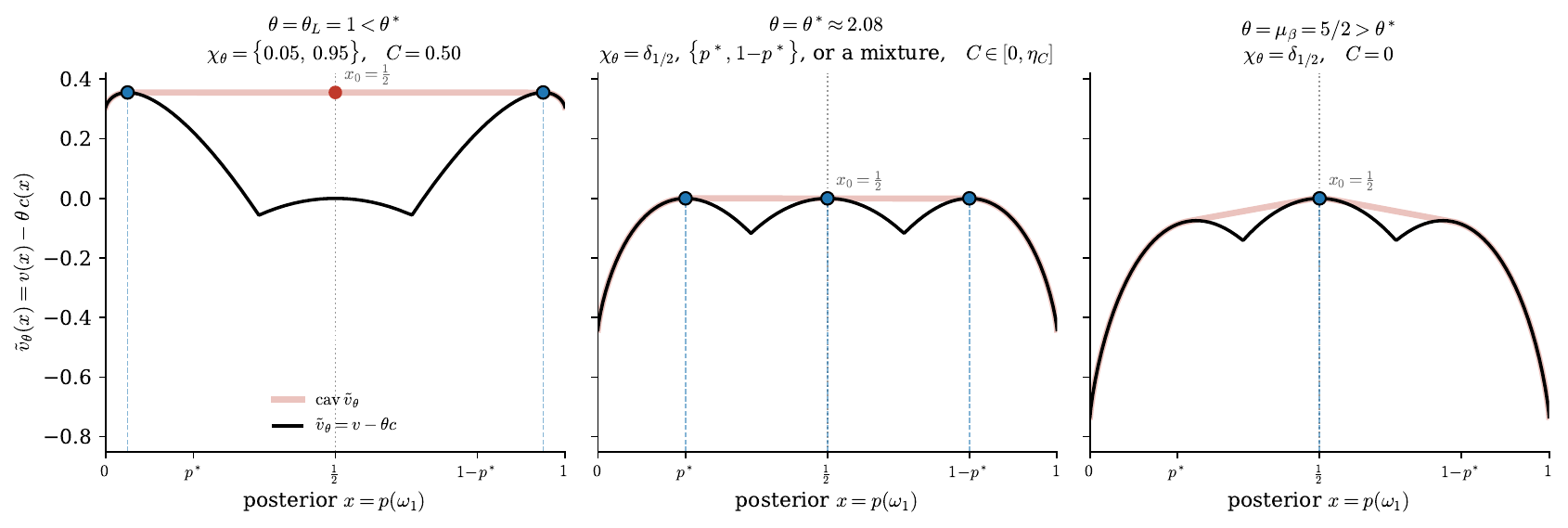}
\caption{
Concavification in the Example. Black curves show $\tilde v_\theta(x)=v(x)-\theta c(x)$, red curves its concave envelope, and blue points the contacts relevant at $p_0=\frac{1}{2}$. Below $\theta^{*}$ the optimum uses two posteriors; above it the optimum is uninformative. At $\theta^*$, the uninformative experiment, the two-posterior experiment with support $\{p^*,1-p^*\}$, and their mixtures are optimal, with costs spanning $[0,\eta_C]$.}
\label{fig:concav}
\end{figure}

For each type $\theta$, the unconstrained problem $\max_\tau\{V(\tau)-\theta C(\tau)\}$ is solved by concavifying the multiplier-adjusted value $\tilde v_\theta=v-\theta c$ over the prior $p_0=\frac{1}{2}$, as illustrated in Figure~\ref{fig:concav}. Because $v$ is flat around the prior $p_0$, where a safe action is optimal, a little information there changes no action yet is still costly, and hence is never optimal. The concave envelope meets $\tilde v_\theta$ either at $p_0$ alone (no information) or at a symmetric pair $\{p_{\theta},1-p_{\theta}\}$ reaching past the flat region. At the knife-edge $\theta^{*}$, the optimal support jumps from $\{p_0\}$ to $\{p^{*},1-p^{*}\}$, and its cost jumps from $0$ to $\eta_{C}=H(\frac{1}{2})-H(p^{*})$. Thus information acquisition is \emph{lumpy}: for every $\theta\neq \theta^{*}$, the optimal cost is $0$ or at least $\eta_{C}$, never an amount strictly in between (Appendix~\ref{app:lumpy}).

This lumpiness carries directly over to the four reference costs that anchor the contract, each $0$ or at least $\eta_{C}$. Recall that $\CFBH\ge\CFBL\ge\CaL\ge\CbL\ge0$ and $\CaL>0$. Hence $\CFBH$, $\CFBL$, $\CaL\geq \eta_{C}$ and $\CbL\in\{0\}\cup[\eta_{C},\infty)$. The low type's post-investment screening benchmark cost $\CbL$ is therefore the only reference cost that can lie below the minimum efficient scale, and if it does, it must collapse all the way to zero. The condition $\CbL=0$ means that the screening benchmark under the post-investment distribution assigns no information to the low type. It does not describe the actual inducing contract, whose investment incentives can require positive low-type information.

This cost gap now delivers what we call the \emph{multi-posterior characterization}: an exact condition on the investment technology (varying $\beta$, hence $\CbL$) for the low type's experiment to use more posteriors than there are states.

\begin{claimx}[Multi-posterior characterization]
\label{claim:multipost}
Hold all primitives of the Example fixed except $\beta\in(\alpha,1)$ and the investment cost $\phi$. Restrict attention to $\beta$ satisfying the standing differentiability assumption.

Three posteriors are necessary for the low type's experiment in an optimal contract on a nonempty open interval of $\phi$ if and only if $\CbL=0$. When $\CbL=0$, they are necessary for every
\[
  0<\phi<\phi_C,
  \qquad
  \phi_C:=(\beta-\alpha)(\thL-\thH)\eta_C.
\]

When $\CbL>0$, an optimal contract can always be selected with at most two low-type posteriors.
\end{claimx}
\begin{proof}
    See Appendix~\ref{app:multipost}.
\end{proof}

In the Example, the multi-posterior regime arises precisely when the post-investment screening benchmark assigns no information to the low type. Investment incentives nevertheless require positive low-type information over a range of investment costs. The resulting cost cap makes a third posterior indispensable, even though the unconstrained information-design problem always admits an optimum with at most two posteriors.

\paragraph{Generalization.} The next result extends beyond the Example. Under the stated conditions, every optimal contract requires the low type to use at least three posteriors despite there being only two states, over a nonempty interval of investment costs.

\begin{proposition}[Investment induces extra posteriors]
\label{prop:mes3}
Let $|\Omega|=2$ and $c(p)=H(p_0)-H(p)$. Suppose:
\begin{enumerate}[label=\textup{(\roman*)}]
  \item \emph{(Safe action at the prior.)} There exist actions $a_0,a_1$, each uniquely optimal at its corresponding degenerate belief, but neither optimal at $p_0$.
  \item \emph{(Decisive low type.)} A $\muA$-optimal experiment has two posteriors, one where $a_0$ is optimal and one where $a_1$ is optimal.
  \item \emph{(Silenced low type.)} $\CbL=0$.
\end{enumerate}

On some nonempty open interval of investment costs, investment is strictly preferred and every optimal contract assigns the low type at least three posteriors.
\end{proposition}

\begin{proof}
See Appendix~\ref{app:mes3}.
\end{proof}

We show that this can occur when the decision problem has at least three actions, and a sufficiently interior prior. Although these conditions are not stated directly in terms of primitives, they are easy to verify in any given contracting problem: The safe-action condition~(i) can be read directly from the decision problem, while decisiveness~(ii) and silencing~(iii) concern the unconstrained optima at the rent-inflated multipliers $\muA$ and $\muB$: the $\muA$-optimum is sufficiently informative, whereas the $\muB$-optimum is uninformative ($\CbL=0$). Each condition can therefore be checked with a single concavification.

\begin{remark}[Affinely dependent support and implementation]
\label{rem:implement}
In the Example, investment incentives make three posteriors necessary for the low type's optimal experiment when $0<\eta<\eta_C$. Its support $\{p^{*},\tfrac12,1-p^{*}\}$ is affinely dependent, specifically, $\tfrac12=\tfrac12 p^{*}+\tfrac12(1-p^{*})$. Although unconstrained information design admits an optimum with at most as many posteriors as states, the investment-induced cost constraint can make an additional posterior indispensable. This is consistent with the role of constraints in expanding the support required by optimal information structures \citep{azrieli2021constrained,letreust2019persuasion,doval2024constrained}.

The necessity of three posteriors also matters when comparing direct contracting on experiments with posterior-contingent payments. Suppose the researcher freely chooses an experiment under a fixed payment rule $t(p)$. His objective, $\mathbb{E}_{\tau}[t(p)-\theta c(p)]$, is linear in the posterior law. An experiment with affinely dependent support can be expressed as a nontrivial mixture of distinct Bayes-plausible experiments. If it is optimal, each component must therefore also be optimal. Consequently, such a payment rule cannot uniquely implement the prescribed experiment, although weak implementation may be possible.

Equal payments across the prescribed support create a stronger obstacle in our binary-state setting. Strict convexity of the acquisition cost allows the researcher to shift probability from the two outer posteriors toward the middle posterior, preserving Bayes plausibility and expected payment while strictly reducing cost. Such payments therefore cannot even weakly implement the three-posterior experiment. In our methods-based model, by contrast, payment is conditional on delivering the contracted experiment, so this deviation does not preserve the researcher's payment.
\end{remark}

\subsection{Prior sensitivity}\label{sec:compstat}
We next examine how posterior support responds to changes in the state prior. Hold the decision problem, cost types, and investment technology fixed, and let the binary-state prior $p_0$ vary. Write
\[
  C_{p_0}(\tau)
  :=H(p_0)-\mathbb{E}_{\tau}[H(p)],
  \qquad \mathbb{E}_{\tau}[p]=p_0.
\]

The relevant distinction is between a problem with a fixed effective cost multiplier and one with a binding information-cost cap. The following result concerns posterior support, not invariance of the entire contract.

\begin{proposition}[Prior sensitivity]\label{prop:compstat}
With mutual-information cost:
\begin{enumerate}[label=\arabic*.]
  \item \emph{(Fixed effective multiplier.)} Fix $\lambda>0$ and write $\tilde v_\lambda=v+\lambda H$, omitting the prior-dependent constant $-\lambda H(p_0)$ from the adjusted value. Suppose $a<b$ are contact points of $\operatorname{cav}(\tilde v_\lambda)$ with $\tilde v_\lambda$, and $\operatorname{cav}(\tilde v_\lambda)$ is affine on $[a,b]$. For every prior $p_0\in(a,b)$, an optimizer of
  \[
    \max_{\tau:\,\mathbb{E}_{\tau}[p]=p_0}
       \{V(\tau)-\lambda C_{p_0}(\tau)\}
  \]
  can be selected with support $\{a,b\}$ and probabilities
  \[
    \tau(a)=\frac{b-p_0}{b-a},
    \qquad
    \tau(b)=\frac{p_0-a}{b-a}.
  \]
  Thus this optimal selection has fixed support as the prior varies within $(a,b)$.
  \item \emph{(Binding information-cost cap.)} Fix $\eta>0$. Suppose that on a nonempty open interval $I\subset(0,1)$, the low type's optimal experiment is two-point and its information-cost cap binds: $C_{p_0}(\tau_{p_0})=\eta$ for every $p_0\in I$. No selection of these optimal experiments can have the same pair of support posteriors throughout $I$.
\end{enumerate}
\end{proposition}
\begin{proof}
For part 1, write $V(\tau)-\lambda C_{p_0}(\tau) =\mathbb{E}_{\tau}[\tilde v_\lambda(p)]-\lambda H(p_0)$. The last term does not depend on the experiment. For any Bayes-plausible $\tau$,
\[
  \mathbb{E}_{\tau}[\tilde v_\lambda(p)]
  \le
  \mathbb{E}_{\tau}[\operatorname{cav}(\tilde v_\lambda)(p)]
  \le
  \operatorname{cav}(\tilde v_\lambda)(p_0).
\]

The stated two-point law attains this bound, because $a$ and $b$ are contact points and the concave envelope is affine between them. It is therefore optimal for every $p_0\in(a,b)$.

For part 2, suppose instead that the support remains $\{a,b\}$ throughout $I$, with $a<b$. Bayes plausibility determines the probabilities, so
\[
  C_{p_0}(\tau_{p_0})
  =
  H(p_0)
  -\frac{b-p_0}{b-a}H(a)
  -\frac{p_0-a}{b-a}H(b).
\]

The last two terms are affine in $p_0$, whereas $H$ is strictly concave on $(0,1)$. The displayed expression is therefore strictly concave and cannot equal the fixed constant $\eta$ throughout an open interval. This contradicts the binding-cap assumption.
\end{proof}

In the screening benchmark without investment opportunities, the relevant multipliers are $\lambda=\theta_H$ for the high type and $\lambda=\mu_\alpha$ for the low type. More generally, under the fixed type distribution in \eqref{eq:Pf}, the low type's multiplier is $\theta_L+\frac{f}{1-f}(\theta_L-\theta_H)$. Part 1 applies when the corresponding concave envelope has the stated affine segment. It does not assert support invariance in a no-information region, where the optimal experiment may be $\delta_{p_0}$. Even when support remains fixed, probabilities, information costs, and transfers may change with the prior.

In the Example, at $p_0=1/2$, the low type's optimum is two-point with a binding cap for $\eta_C<\eta<C_L^{FB}$. Part 2 describes its response to prior changes conditional on this regime persisting on an open interval of priors.

\section{Extension: More than Two Types}\label{sec:ntype}

We now allow for any finite number of cost types. Two properties survive: the principal's screening value is non-decreasing under FOSD improvements in the type distribution, and her investment decision retains a cutoff in the sunk cost. The welfare comparison changes, however. While the two-type model rules out overinvestment, a three-type example shows that the principal may induce investment that is inefficient relative to the first best.

\subsection{Framework}\label{sec:ntype-frame}
We keep the arbitrary finite state space $\Omega$, the common prior $p_0\in\Delta(\Omega)$, and the information-design primitives $\X,C,V$ of Section \ref{sec:model}. For finite $\Omega$, the set $\X$ is compact and convex, $C$ and $V$ are affine and bounded on $\X$ with $\bar C:=\max_{\tau\in\X}C(\tau)<\infty$, and each program $\max_{\tau\in\X}\{V(\tau)-\lambda C(\tau)\}$ attains its maximum. Every argument below runs on the scalar costs $C_i:=C(\chi_i)$, the values $V(\chi_i)$, and the multipliers.

Only the type space now expands. Let there be $n\ge2$ types $\theta_1>\dots>\theta_n$, with type $n$ the \emph{most efficient} (lowest marginal cost of information), and let $g=(g_1,\dots,g_n)$, $g_i>0$, $\sum_i g_i=1$, be the pre-investment type distribution, with cdf $G_k:=\sum_{i\ge k}g_i$ (so $G_{n+1}=0$, $G_1=1$).

A methods-based contract $(\chi,T)$ contains $\chi_i:=\chi(\theta_i)\in\X$ and $T_i:=T(\theta_i)\in\R_+$. The feasible set
\[
  \mathcal P_n=\Bigl\{(\chi,T):
    \begin{array}{l}
      T_i-\theta_i C_i\ \ge\ T_j-\theta_i C_j\ \ \forall\, i,j,\\[2pt]
      T_i-\theta_i C_i\ \ge\ 0\ \ \forall\, i
    \end{array}\Bigr\}
\]
is cut out by \eqref{eq:IC} and \eqref{eq:IR} alone, so it \emph{does not depend on $g$}.

The screening value is $V_g=\max_{(\chi,T)\in\mathcal P_n}\sum_{i=1}^n g_i\,[V(\chi_i)-T_i]$, the $n$-type analog of $(P_f)$ of Appendix~\ref{app:obs} (equivalently of \eqref{eq:PNI} at $\eta\to\infty$). For each type record its rent $U_i:=T_i-\theta_i C_i\ge0$, total surplus $W_i:=V(\chi_i)-\theta_i C_i$, and the principal's net payoff $S_i:=V(\chi_i)-T_i=W_i-U_i$.

Throughout, we impose full support as a standing assumption: $g_i,g_i'>0$ for every type $i$, for both the pre-investment distribution $g$ and any post-investment distribution $g'$ introduced from Definition~\ref{def:fosd} onward.

\begin{definition}[Virtual multipliers; regularity]\label{as:regular}
Given $g$, define the \emph{virtual (rent-inflated) multipliers} by
\[
  \lambda_n:=\theta_n,\qquad
  \lambda_i:=\theta_i+\frac{G_{i+1}}{g_i}\,(\theta_i-\theta_{i+1})\quad(1\le i\le n-1).
\]

Each $\lambda_i\ge\theta_i$ inflates type $i$'s marginal cost $\theta_i$ by the information rent conceded to the more efficient types $i+1,\dots,n$; at the most efficient type $\lambda_n=\theta_n$ (no distortion). We call the type distribution $g$ \emph{regular} if its virtual multipliers are non-increasing, i.e., $\lambda_1\ge\lambda_2\ge\dots\ge\lambda_n$.
\end{definition}

\begin{definition}[FOSD investment shift]\label{def:fosd}
We say that $g'$ FOSD-dominates $g$ toward efficient types if
\[
  G'_k\ge G_k \qquad \text{for every }k,
\]
where $G_k=\sum_{i\ge k}g_i$ and $G'_k=\sum_{i\ge k}g'_i$.
\end{definition}

We now reinstate the hidden-investment margin of Section \ref{sec:model}. Investment shifts the type distribution from $g$ to $g'$ at sunk cost $\phi>0$, where $g'$ FOSD-dominates $g$ toward efficient types. The principal posts one contract $(\chi,T)\in\mathcal P_n$; the researcher chooses whether to invest and then, after his type realizes, reports it.

\begin{definition}[$n$-type investment-incentive constraint]\label{def:invest-ic}
Fix a contract $(\chi,T)\in\mathcal P_n$ and let $U_i=T_i-\theta_i C_i\ge0$ (\eqref{eq:IR}) be the interim rent of type $i$ --- which, since the reported type faces its own \eqref{eq:IC}-optimal option, is exactly type $i$'s continuation payoff. The researcher's ex-ante expected utility is $\sum_i g_i U_i$ if he does not invest and $\sum_i g'_i U_i-\phi$ if he does. Hence, offered $(\chi,T)$, \emph{the researcher is willing to invest if and only if}
\begin{equation}
  \langle g'-g,\,U\rangle\ :=\ \sum_{i=1}^n(g'_i-g_i)\,U_i\ \ge\ \phi .
  \tag{IC-I$_n$}\label{eq:invest-ic}
\end{equation}

For every feasible contract, summation by parts gives
\begin{equation}
  \langle g'-g,U\rangle
  =\sum_{k=2}^n(G'_k-G_k)(U_k-U_{k-1}).
  \tag{$\dagger$}\label{eq:abel-wedge}
\end{equation}

Investment incentives therefore depend on a weighted combination of adjacent rent gaps. Under FOSD, the weights are nonnegative, although some may be zero. This identity does not require rents to be minimal.
\end{definition}

To induce investment the principal chooses a contract in $\mathcal P_n$ satisfying \eqref{eq:invest-ic}, whereupon types are drawn from $g'$ and she collects $\sum_i g'_i S_i$; to deter it she chooses a contract with the reverse (weak) inequality and collects $\sum_i g_i S_i$ under $g$. Writing $S_i=V(\chi_i)-T_i$ as before, define the \emph{induce} and \emph{deter} values
\begin{align}
  V_I(\phi)&=\max\Bigl\{\textstyle\sum_{i=1}^n g'_i S_i:\ (\chi,T)\in\mathcal P_n,\
    \langle g'-g,U\rangle\ge\phi\Bigr\}, \label{eq:VI-ntype}\\[2pt]
  V_{NI}(\phi)&=\max\Bigl\{\textstyle\sum_{i=1}^n g_i S_i:\ (\chi,T)\in\mathcal P_n,\
    \langle g'-g,U\rangle\le\phi\Bigr\}, \label{eq:VNI-ntype}
\end{align}
and $D(\phi)=V_I(\phi)-V_{NI}(\phi)$; the principal prefers investment if and only if $D(\phi)\ge0$. Both feasible sets carve $\mathcal P_n$ by the single scalar constraint \eqref{eq:invest-ic}, and both objectives are $\phi$-free.

For a nonregular type distribution, the unconstrained screening benchmark $V_g$ is characterized by ironing the virtual multipliers (Lemma~\ref{lem:iron}, Appendix~\ref{app:ntype-pf}). This characterization does not by itself solve the investment-inducing or investment-deterring programs, whose additional constraints can change both rents and allocations.

\subsection{Main results}\label{sec:ntype-results}

\begin{proposition}[Properties of the screening value]\label{prop:Vf-ntype}
For type distributions $g$ and $g'$:
\begin{enumerate}[label=\arabic*.]
  \item \emph{(Convexity.)} For every $t\in[0,1]$, $V_{(1-t)g+tg'}\le (1-t)V_g+tV_{g'}$.
  \item \emph{(Monotonicity.)} If $g'$ FOSD-dominates
  $g$ toward more efficient types, then $V_{g'}\ge V_g$.
\end{enumerate}
For $n=2$, these conclusions recover
Lemma~\ref{lem:Vf}.
\end{proposition}

\begin{proof}
See Appendix~\ref{app:gen-vf-pf}.
\end{proof}

The principal's screening value weakly increases when the type distribution improves. To see why, select the optimal screening contract constructed by Lemma~\ref{lem:iron}, with coordinated experiment choices and minimal rents. Its net payoffs satisfy $S_1\le\cdots\le S_n$. The same contract remains feasible under $g'$, so
\[
  V_{g'}-V_g
  \ge \langle g'-g,S\rangle
  =\sum_{k=2}^n(G'_k-G_k)(S_k-S_{k-1})
  \ge0.
\]

The final inequality follows from FOSD. The construction applies to both regular and nonregular distributions.

\begin{proposition}[Investment cutoff]\label{prop:cutoff-ntype}
Fix distributions $g,g'$ and sunk cost $\phi>0$, with $V_I,V_{NI}$ as in \eqref{eq:VI-ntype}--\eqref{eq:VNI-ntype}.
\begin{enumerate}[label=\arabic*.]
  \item \emph{(Interval.)} $D=V_I-V_{NI}$ is non-increasing in $\phi$ and negative for $\phi>\bar\phi:=(\theta_1-\theta_n)\bar C$. The set $\Phi_D:=\{\phi\geq0:D(\phi)\geq0\}$ is either empty or a closed interval $[0,\phi_I]$, with $\phi_I\leq \bar{\phi}$.
  \item \emph{(Non-emptiness.)} If $g'$ FOSD-dominates $g$ toward efficient types, then $V_I(0)=V_{g'}$ and $V_{NI}(0)\le V_g\le V_{g'}$ (the last inequality by Proposition~\ref{prop:Vf-ntype}), so $D(0)=V_{g'}-V_{NI}(0)\ge0$: the interval is non-empty, and investment is \emph{strictly} preferred on some $[0,\epsilon)$ whenever $V_{g'}>V_{NI}(0)$.
  \item \emph{(Two-type case.)} For $n=2$, the result reduces to Proposition~\ref{prop:cutoff}.
\end{enumerate}
\end{proposition}

\begin{proof}
See Appendix~\ref{app:gen-cutoff-pf}.
\end{proof}

The interval shape (Part~1) is pure constraint nesting and needs only the standing primitives. Indeed, tightening the induce floor \eqref{eq:invest-ic} shrinks the induce-feasible set and enlarges the deter-feasible set, so $D$ falls monotonically. What FOSD adds, and all it adds, is \emph{non-emptiness} (Part~2): value monotonicity (Proposition~\ref{prop:Vf-ntype}) gives $D(0)\ge0$.

With two types, the single rent gap and the reporting constraints yield the cost thresholds characterized in Lemma~\ref{lem:PI}. With more types, investment incentives can depend on several adjacent rent gaps. In general, the investment constraint does not reduce to a cost threshold on a single assigned experiment, although special shifts may place positive weight on only one gap.

The cutoff property extends to finitely many types, but the two-type no-overinvestment conclusion need not. Define the first-best cutoff by
\begin{equation}
 \phi^{FB}_n:=\sum_{i=1}^n(g'_i-g_i)K_{\theta_i},
 \qquad K_{\theta_i}:=\max_{\tau\in\X}\{V(\tau)-\theta_i C(\tau)\}.
 \label{eq:phiFBn}
\end{equation}

This is the increase in maximal expected total surplus before subtracting investment cost.

\begin{proposition}[Overinvestment with three types]\label{prop:overinvestment}
There exists a binary-state decision problem with mutual-information costs, three cost types, and full-support distributions $g,g'$ such that $g'$ FOSD-dominates $g$ toward efficient types and $\phi_I>\phi^{FB}_3$.
\end{proposition}
\begin{proof}
See Appendix~\ref{app:overinvestment}.
\end{proof}

This proposition shows that the principal can strictly prefer inducing investment even when its cost exceeds the resulting gain in first-best expected surplus. The contrast with the two-type model lies in how rents provide investment incentives. Normalize the least efficient type's rent to zero. With two types, one rent gap determines both the researcher's gain from investment and the principal's expected rent payment. Together with the available deterrence contract, this restriction rules out overinvestment. With more types, investment incentives weight adjacent rent gaps by $G'_k-G_k$, while expected rent payments after investment weight them by $G'_k$. Different gaps can therefore provide different amounts of investment incentive per unit of expected rent payment. The example relies mainly on the gap between the two most efficient types, making investment more profitable than deterrence even above the first-best investment threshold.

\section{Concluding Remarks}\label{sec:conclude}
This paper studies contracting for information when a researcher makes a hidden investment before learning his privately observed cost type. With two cost types, the contracting problem reduces to constrained information design. The principal may underprovide investment but never induces investment beyond the first-best threshold. With finitely many types, the investment cutoff survives, while overinvestment becomes possible already with three types. Thus, the cutoff structure alone does not determine whether private investment incentives are excessive or insufficient.

Investment incentives also reshape the experiments assigned to researchers. In the binary-state example, every optimal low-type experiment requires three posteriors over a range of investment costs. This affinely dependent support matters for implementation: a fixed posterior-contingent payment rule cannot uniquely implement the prescribed experiment. Equal payments across its support create a stronger obstacle, allowing a strictly cheaper deviation and therefore ruling out even weak implementation.

\paragraph{Two assumptions.} We now discuss two important assumptions of the model. First, the experiment is contractible. This is a strong assumption, but a natural one in settings where experimental design and results are observable and verifiable. Clinical trials, for example, are subject to preregistration and disclosure requirements. Because contractibility favors the principal, our analysis can be interpreted as an upper bound on what she can achieve. \citet{yoder2022designing} and \citet{wang2025contracting} also study methods-based contracting. When the experiment cannot be contracted on directly, the principal must instead rely on coarser instruments, for example, conditioning payments on the realized state or the experiment result, or, with several researchers, having them monitor one another.

Second, the principal commits before the researcher invests, and investment precedes the realization of his type. This timing applies to settings where a project-specific, cost-reducing investment changes the \emph{distribution} of the researcher's future cost rather than responding to a cost already realized. The principal can influence investment incentives through the ex ante contract. If the researcher received a private signal about his future cost before investing, the contract would also need to account for this initial private information. Investment could nevertheless improve the distribution of his subsequently realized cost.

\paragraph{Takeaways.} Two broader lessons stand out. First, an improvement in the researcher's cost distribution raises the principal's screening value, but this value comparison alone does not determine whether inducing the improvement is socially desirable. The principal compares payoffs net of transfers, while the first-best benchmark compares total surplus net of investment cost. The three-type example shows that these comparisons can favor different investment decisions even under full commitment.

Second, investment incentives can complicate implementation by making affinely dependent posterior support necessary. They can also affect prior sensitivity: when a cost cap binds throughout a two-posterior regime, the same support cannot remain optimal over an open interval of priors.

\newpage
\appendix

\section{The First-Best Benchmark}\label{app:FB}
The first best eliminates all contracting frictions. Conditional on each cost type, the efficient experiment choice function $\chi^{FB}$ satisfies
\[
  \chi^{FB}(\thH)\in\argmax_{\tau\in\X}\{V(\tau)-\thH C(\tau)\},\qquad
  \chi^{FB}(\thL)\in\argmax_{\tau\in\X}\{V(\tau)-\thL C(\tau)\}.
\]

\section{The Second-Best Benchmark}\label{app:obs}

Fix the probability $f$ of the high type and set aside the investment decision. The researcher privately observes his realized cost type. The principal's screening payoff, subject to reporting incentives and interim participation, is
\begin{equation}
  V_f=\max_{(\chi,T)} f[V(\chi(\thH))-T(\thH)]+(1-f)[V(\chi(\thL))-T(\thL)]
  \quad\text{s.t. \eqref{eq:IC}, \eqref{eq:IR}.} \tag{$P_f$}\label{eq:Pf}
\end{equation}

The solution is characterized by \citet{yoder2022designing}. It is the $\eta\to\infty$ special case of Lemma~\ref{lem:PNI}, where the cost constraint becomes vacuous. Write $V_\alpha:=V_f|_{f=\alpha}$, $W_f(\chi,T)=f[V(\chi(\thH))-T(\thH)]+(1-f)[V(\chi(\thL))-T(\thL)]$, and $\mathcal P=\{(\chi,T):\text{\eqref{eq:IC},\eqref{eq:IR}}\}$, so $V_f=\max_{\mathcal P}W_f$.

\begin{lemma}\label{lem:Vf}
The screening value $V_f$ is convex and non-decreasing in $f$. Moreover, an optimal screening contract $(\chi^f,T^f)$ can be selected such that, for every $f'\ge f$,
\[
  V_{f'}\ge W_{f'}(\chi^f,T^f)\ge V_f.
\]
\end{lemma}
\begin{proof}
The feasible set $\mathcal P$ is independent of $f$, and $W_f(\chi,T)$ is affine in $f$ for each fixed contract. Hence $V_f=\max_{(\chi,T)\in\mathcal P}W_f(\chi,T)$ is convex.

Select an optimal screening contract $(\chi^f,T^f)$ with a first-best high-type experiment and minimal rents. Its transfers satisfy
\[
  T^f(\thL)=\thL C(\chi^f(\thL)),\qquad
  T^f(\thH)=\thH C(\chi^f(\thH))
    +(\thL-\thH)C(\chi^f(\thL)).
\]

For $0<f<1$, these properties follow from the unconstrained screening characterization. At $f=0$, assign each type a first-best experiment and use the displayed transfers; cost monotonicity ensures incentive compatibility. At $f=1$, assign the low type the uninformative experiment and the high type a first-best experiment, again using the displayed transfers. These contracts attain the corresponding endpoint screening values.

For any $f'\ge f$, substitution gives
\begin{equation}
\begin{aligned}
  &W_{f'}(\chi^f,T^f)-W_f(\chi^f,T^f)\\
  &\quad=(f'-f)\Bigl\{
    [V(\chi^f(\thH))-\thH C(\chi^f(\thH))]
    -[V(\chi^f(\thL))-\thH C(\chi^f(\thL))]
  \Bigr\}\ge0,
\end{aligned}
\label{eq:netgain}
\end{equation}
because $\chi^f(\thH)$ maximizes $V-\thH C$. The same contract is feasible under $f'$, so
\[
  V_{f'}\ge W_{f'}(\chi^f,T^f)
  \ge W_f(\chi^f,T^f)=V_f.
\]

Thus $V_f$ is non-decreasing.
\end{proof}

Under the rent-minimizing screening contract, the efficient type receives his first-best experiment and earns the rent $(\thL-\thH)C(\chi(\thL))$. This rent equals the cost saving he would enjoy from running the inefficient type's experiment. After paying it, the principal's payoff advantage from the efficient type is therefore
\[
  \bigl[V(\chi(\thH))-\thH C(\chi(\thH))\bigr]
  -\bigl[V(\chi(\thL))-\thH C(\chi(\thL))\bigr]\ge0.
\]

The remaining advantage is the surplus gain from assigning the efficient type his own optimal experiment. Consequently, a higher probability of the efficient type weakly raises the principal's payoff from this contract.

The next result shows the conclusion survives when investment is a hidden action: her equilibrium payoff never falls below the benchmark $V_\alpha$ in which the researcher is exogenously a high type with probability $\alpha$. Here, $V_I(\phi),V_{NI}(\phi)$ are the values of \eqref{eq:PI},\eqref{eq:PNI}.

\begin{proposition}[Investment is never a liability]\label{prop:noliability}
For every $\phi\ge0$, $\max\{V_I(\phi),V_{NI}(\phi)\}\ge V_\alpha$.
\end{proposition}
\begin{proof}
Select an optimal screening contract $(\chi^\alpha,T^\alpha)$ under $\alpha$, with rents
\[
  U_{\thL}=0,\qquad
  U_{\thH}=(\thL-\thH)\CaL.
\]

If $\eta\ge\CaL$, this contract satisfies the deterrence constraint, so $V_{NI}(\phi)\ge W_\alpha(\chi^\alpha,T^\alpha)=V_\alpha$. If $\eta<\CaL$, it satisfies the investment-inducing constraint, so $V_I(\phi)\ge W_\beta(\chi^\alpha,T^\alpha)\ge W_\alpha(\chi^\alpha,T^\alpha)=V_\alpha$, where the second inequality follows from \eqref{eq:netgain}.

Therefore, $\max\{V_I(\phi),V_{NI}(\phi)\}\ge V_\alpha$ for every $\phi\ge0$.
\end{proof}

\begin{remark}[Investment incentives under commitment]
\label{rem:holdup}
The investment opportunity never reduces the principal's optimal payoff below $V_\alpha$. This benchmark already leaves information rent to the high type. Inducing investment requires the rent difference between types to provide sufficient incentives, potentially changing both transfers and assigned experiments. With two types, the resulting payoff comparison rules out overinvestment, although underinvestment can occur. This distortion arises under full commitment and involves no ex post revision of contractual terms.
\end{remark}

\section{Omitted Proofs}\label{app:omitted}

\subsection{Proof of Lemma~\ref{lem:blackwell} (Blackwell monotonicity)}\label{app:blackwell}
\begin{proof}
Identify a posterior with $x\in[0,1]$ and write $x_0$ for the prior. Both experiments have mean $x_0$. By the concavification characterization \citep{kamenica2011bayesian}, for each $i$, there is an affine supporting majorant $L_i$ of $v-\theta_i c$ satisfying $L_i(x_0)=\operatorname{cav}(v-\theta_i c)(x_0)$. Optimality implies $\operatorname{supp}\tau_i \subseteq \{x:v(x)-\theta_i c(x)=L_i(x)\}$.

Set $\delta=\theta_1-\theta_2>0$ and define
\[
  Q(x):=L_2(x)-L_1(x)-\delta c(x).
\]

The majorant inequalities and contact equalities give
\[
  Q\ge0\quad\text{on }\operatorname{supp}\tau_1,
  \qquad
  Q\le0\quad\text{on }\operatorname{supp}\tau_2.
\]

Since $c$ is strictly convex, $Q$ is strictly concave. Its nonempty superlevel set $\{x\in[0,1]:Q(x)\ge0\}$ is therefore a closed interval $[a,b]$ containing $\operatorname{supp}\tau_1$ and hence $x_0$.

If $a=b$, then $\tau_1=\dirac_{x_0}$, so $\tau_2\succeq_B\tau_1$. Otherwise, strict concavity gives $Q>0$ on $(a,b)$, and consequently
\[
  \operatorname{supp}\tau_1\subseteq[a,b],
  \qquad
  \operatorname{supp}\tau_2
  \subseteq[0,a]\cup[b,1].
\]

For any convex function $\psi$, let $\ell$ be the affine function agreeing with $\psi$ at $a$ and $b$. Convexity implies $\psi\le\ell$ on $[a,b]$ and $\psi\ge\ell$ outside $(a,b)$. Thus
\[
  \mathbb E_{\tau_1}[\psi]
  \le \mathbb E_{\tau_1}[\ell]
  =\ell(x_0)
  =\mathbb E_{\tau_2}[\ell]
  \le \mathbb E_{\tau_2}[\psi].
\]

Thus $\tau_2$ dominates $\tau_1$ in convex order. For distributions of posterior beliefs with the same prior, this is equivalent to Blackwell dominance \citep{blackwell1951comparison,blackwell1953equivalent}, so $\tau_2\succeq_B\tau_1$.

Finally, if $C(\tau_2)>C(\tau_1)$, the two posterior distributions differ. Antisymmetry of convex order then makes the dominance strict.
\end{proof}

\subsection{Auxiliary results for two states}\label{app:capped}
\begin{lemma}[Uniqueness at a differentiability point]
\label{lem:binary-unique}
Under the maintained finite-action and strictly convex cost assumptions, let $|\Omega|=2$ and $\lambda>0$. If $K$ is differentiable at $\lambda$, then
\[
  \argmax_{\tau\in\X}\{V(\tau)-\lambda C(\tau)\}
\]
contains a unique distribution over posterior beliefs.
\end{lemma}

\begin{proof}
Differentiability implies that every optimizer has cost $-K'(\lambda)$. Let $L$ be an affine supporting majorant of $v-\lambda c$ attaining its concavification at the prior. Every optimizer is supported on the contact set $\mathcal{S}=\{p:v(p)-\lambda c(p)=L(p)\}$.

Write $v_a(p):=\sum_\omega p(\omega)u(a,\omega)$. The contact set is finite. For each action $a$, the function $v_a-\lambda c-L$ is strictly concave and nonpositive, so it can vanish at at most one point. Every contact is a zero of this function for some action, and there are finitely many actions.

Suppose an optimizer assigns positive probability to three distinct contacts $p_1<p_2<p_3$. Write $p_2=q p_1+(1-q)p_3$, with $q\in(0,1)$. Moving sufficiently small masses $\varepsilon q$ and $\varepsilon(1-q)$ from $p_1$ and $p_3$ to $p_2$ preserves Bayes plausibility. Since all three points are contacts with the same affine majorant, the modified experiment remains optimal. But strict convexity of $c$ strictly lowers its cost, contradicting the common optimal cost. Thus every optimizer has at most two support points.

If two distinct optimizers existed, their equal mixture would also be optimal. Its support would contain at least three points, since probabilities on a fixed set of at most two points are uniquely determined by the prior. This is a contradiction.
\end{proof}
\begin{lemma}[Cost-capped optima are Blackwell-dominated]
\label{lem:capped}
Let $|\Omega|=2$, $\theta'\ge\theta>0$, and $\eta\ge0$. Suppose $K$ is differentiable at $\theta$, and let $\tau_\theta^*$ be the unique unconstrained $\theta$-optimizer. Every solution $\tau$ to
\[
  \max_{\tau\in\X:\,C(\tau)\le\eta}
  \{V(\tau)-\theta'C(\tau)\}
\]
satisfies $\tau\preceq_B\tau_\theta^*$, with strict Blackwell dominance whenever $C(\tau)<C(\tau_\theta^*)$.
\end{lemma}

\begin{proof}
If $\eta=0$, strict convexity of the posterior cost implies that $\tau$ is uninformative, so the result follows. Suppose $\eta>0$.

The uninformative experiment is strictly feasible. Strong duality therefore gives a multiplier $m\ge0$ such that every capped optimizer also maximizes $V-(\theta'+m)C$ over $\X$. Write $\lambda=\theta'+m\ge\theta$.

If $\lambda>\theta$, Lemma~\ref{lem:blackwell} gives $\tau\preceq_B\tau_\theta^*$, strictly whenever $C(\tau)<C(\tau_\theta^*)$. If $\lambda=\theta$, then $\tau$ is an unconstrained $\theta$-optimizer and hence equals $\tau_\theta^*$ by Lemma~\ref{lem:binary-unique}.
\end{proof}

\subsection{Proof of Lemma~\ref{lem:PNI}}\label{app:PNI}
\begin{proof}
Write $C_H=C(\chi(\thH))$ and $C_L=C(\chi(\thL))$. Reporting and deterrence incentives require $(\thL-\thH)C_L\le U_H-U_L\le (\thL-\thH)C_H$, and $U_H-U_L\le(\thL-\thH)\eta$. These inequalities imply $C_L\le\eta$ and $C_L\le C_H$.

For any allocation satisfying these restrictions, the principal minimizes transfers by setting $U_L=0$ and $U_H=(\thL-\thH)C_L$. These rents satisfy participation and preserve both reporting and deterrence incentives. Since both types have positive probability, every optimum uses these rents. Substituting gives the reduced program
\[
  \max_{\chi}\;
  \alpha[V(\chi(\thH))-\thH C_H]
  +(1-\alpha)[V(\chi(\thL))-\muA C_L],
  \qquad
  C_L\le\eta,\quad C_L\le C_H.
\]

Temporarily drop $C_L\le C_H$. The two maximizations then separate. Every high-type optimizer has cost $\CFBH$, and every unconstrained $\muA$-optimizer has cost $\CaL$. If $\eta\ge\CaL$, the latter optimizers are feasible, so every capped optimizer is also unconstrained-optimal and has cost $\CaL$. If $\eta<\CaL$, every capped optimizer instead has cost at most $\eta<\CaL$. Thus every pair of relaxed optimizers satisfies
\[
  C_L\le\CaL\le\CFBH=C_H,
\]
where the middle inequality follows from Lemma~\ref{lem:mono}. The relaxation is therefore exact, and every optimum solves the two separate maximizations. The transfer formulas follow from the rents above.
\end{proof}

\subsection{Proof of Lemma~\ref{lem:PI}}\label{app:PI}
\begin{proof}
Write $C_H=C(\chi(\thH))$ and $C_L=C(\chi(\thL))$. Reporting incentives require $(\thL-\thH)C_L\le U_H-U_L\le (\thL-\thH)C_H$, while investment incentives require $U_H-U_L\ge(\thL-\thH)\eta$. Together, these inequalities imply $C_H\ge\max\{C_L,\eta\}$.

For any allocation satisfying this restriction, the principal therefore minimizes transfers by setting $U_L=0$ and $U_H=(\thL-\thH)\max\{C_L,\eta\}$. The remaining feasibility condition is $C_H\ge\max\{C_L,\eta\}$.

If $C_L\le\eta$, the principal's objective becomes
\[
  \beta[V(\chi(\thH))-\thH C_H]
  +(1-\beta)[V(\chi(\thL))-\thL C_L]
  -\beta(\thL-\thH)\eta.
\]

The constraints are $C_H\ge\eta$ and $C_L\le\eta$, so the two choices separate as in Case~1.

If $C_L\ge\eta$, the objective becomes
\[
  \beta[V(\chi(\thH))-\thH C_H]
  +(1-\beta)[V(\chi(\thL))-\muB C_L],
\]
subject to $C_H\ge C_L\ge\eta$. Relax the constraints to $C_H,C_L\ge\eta$. Since $\thH<\muB$, the cost-monotonicity argument of Lemma~\ref{lem:mono}, applied on the common feasible set $\{\tau\in\X:C(\tau)\ge\eta\}$, ensures that any pair of relaxed optimizers satisfies $C_H\ge C_L$. The relaxation is therefore exact, giving the two optimization problems in Case~2.

It remains to determine which region is optimal. The high-type problem is identical in both regions.

Suppose first that $\eta\ge\CbL$. Every optimizer of the floored $\muB$-problem has cost exactly $\eta$. At $\eta=\CbL$, this follows because every unconstrained $\muB$-optimizer has cost $\CbL$. If $\eta>\CbL$, an optimizer with cost strictly above $\eta$ could be mixed with an unconstrained $\muB$-optimizer to obtain cost $\eta$ and a strictly higher objective: the former cannot be unconstrained-optimal because its cost differs from $\CbL$.

At cost $\eta$, the low-type contribution in the second region is
\[
  (1-\beta)[V(\tau)-\muB\eta]
  =(1-\beta)[V(\tau)-\thL\eta]
    -\beta(\thL-\thH)\eta.
\]

This is also its contribution in the first region, where the same experiment is feasible. Thus the first region attains at least as much value. Moreover, any global optimum in the second region has $C_L=\eta$ and also solves the first region. Hence every optimum has the characterization in Case~1.

Now suppose that $\eta<\CbL$. The second region attains low-type contribution $(1-\beta)K_{\muB}$, since every unconstrained $\muB$-optimizer is feasible. For any experiment with $C(\tau)\le\eta$, the contribution in the first region satisfies
\begin{align*}
  &(1-\beta)[V(\tau)-\thL C(\tau)]
    -\beta(\thL-\thH)\eta\\
  &\quad=
    (1-\beta)[V(\tau)-\muB C(\tau)]
    -\beta(\thL-\thH)[\eta-C(\tau)]\\
  &\quad<(1-\beta)K_{\muB}.
\end{align*}

The strict inequality follows because $C(\tau)\le\eta<\CbL$, whereas every unconstrained $\muB$-optimizer has cost $\CbL$. Since the capped problem attains its maximum, the second region is strictly better. This establishes Case~2.

The stated transfers follow from the rent formulas above.
\end{proof}

\subsection{Proof of Proposition~\ref{prop:dist}}\label{app:props23}
\begin{proof}
\emph{Without investment.} By Lemma~\ref{lem:PNI}, the high type receives a first-best experiment, while the low type solves
\[
  \max_{\tau\in\X:\,C(\tau)\le\eta}
  \{V(\tau)-\muA C(\tau)\}.
\]

Every unconstrained $\muA$-optimizer has cost $\CaL$. If $\eta\ge\CaL$, these optimizers are feasible, so every capped optimizer is also unconstrained-optimal. Otherwise, feasibility gives $C(\chi^*(\thL))\le\eta<\CaL$. Thus
\[
  C(\chi^*(\thL))\le\CaL\le\CFBL.
\]

\emph{With investment.} Consider a feasible inducing program. By Lemma~\ref{lem:PI}, the high type solves
\[
  \max_{\tau\in\X:\,C(\tau)\ge\eta}
  \{V(\tau)-\thH C(\tau)\}.
\]

If $\eta\le\CFBH$, every first-best optimizer is feasible, so every constrained optimizer is first-best. If $\eta>\CFBH$, the floor binds. Indeed, an optimizer with cost strictly above $\eta$ could be mixed with a first-best experiment to reach cost $\eta$, strictly improving its objective. Hence $C(\chi^*(\thH))=\eta>\CFBH$.

For the low type, distinguish the two cases of Lemma~\ref{lem:PI}. If $\eta<\CbL$, every unconstrained $\muB$-optimizer satisfies the floor. Thus every constrained optimizer is $\muB$-optimal and has cost $\CbL\le\CFBL$. If $\eta\ge\CbL$, the low type instead solves
\[
  \max_{\tau\in\X:\,C(\tau)\le\eta}
  \{V(\tau)-\thL C(\tau)\}.
\]

For $\eta<\CFBL$, the cap excludes every first-best optimizer and must bind: any optimizer with slack could be improved by mixing it with a first-best experiment. Therefore $C(\chi^*(\thL))=\eta<\CFBL$. For $\eta\ge\CFBL$, every first-best optimizer is feasible, and every constrained optimizer is first-best.

At $\eta=\CFBL$, the low-type cap binds without distortion; at $\eta=\CFBH$, the high-type floor binds without distortion.

\emph{Blackwell comparisons for two states.} Under \eqref{eq:PNI}, for $\eta>0$ apply Lemma~\ref{lem:capped} with $\theta'=\muA>\theta=\thL$. At $\eta=0$, the cap forces the uninformative experiment, which is Blackwell dominated by every experiment.

Under \eqref{eq:PI}, Case~2 follows from Lemma~\ref{lem:blackwell}, since $\muB>\thL$. In Case~1, for $0<\eta<\CFBL$ apply Lemma~\ref{lem:capped} with $\theta'=\theta=\thL$. At $\eta=0$, the cap again forces the uninformative experiment. For $\eta\ge\CFBL$, the assigned experiment equals the first-best reference experiment by Lemma~\ref{lem:binary-unique}. Thus the low type's experiment is Blackwell dominated by its first-best benchmark in either program.

Finally, at each reference multiplier, all optimizers have the same cost and value. Replacing an assigned first-best optimizer by a fixed first-best representative therefore preserves transfers, payoffs, and all contract constraints. This gives the stated optimal selection for general finite states. For binary states, uniqueness follows from Lemma~\ref{lem:binary-unique}.
\end{proof}

\subsection{Proof of Proposition~\ref{prop:cutoff}}\label{app:cutoff}
\begin{proof}
Write
\[
  D(\phi):=V_I(\phi)-V_{NI}(\phi).
\]

We show that $D$ is non-increasing and upper semicontinuous, with $D(0)\ge0$ and $D(\phi)<0$ for sufficiently large $\phi$.

\emph{Step 1: normalization and feasibility.} Reporting IC gives
\[
  (\thL-\thH)C(\chi(\thL))
  \le U_{\thH}-U_{\thL}
  \le(\thL-\thH)C(\chi(\thH)).
\]

In particular, $U_{\thH}\ge U_{\thL}\ge0$. Subtracting $U_{\thL}$ from both transfers preserves reporting incentives, participation, and the investment wedge, while increasing the principal's payoff by $U_{\thL}$. Thus both programs may be restricted to
\[
  \mathcal Z
  :=\{(\chi,T)\in\mathcal P:U_{\thL}=0\}.
\]

On this domain,
\[
  0\le U_{\thH}\le(\thL-\thH)\bar C,
  \qquad
  0\le T(\theta)\le\thL\bar C.
\]

Under the maintained assumptions, $\X$ is compact and convex, and $V$ and $C$ are continuous and affine. Hence $\mathcal Z$ is compact and convex, and every nonempty induce or deter program attains its maximum.

The investment wedge is at most
\[
  \bar\phi:=(\beta-\alpha)(\thL-\thH)\bar C.
\]
This bound is attained by assigning both types an experiment of cost $\bar C$ at the common payment $\thL\bar C$. Thus the inducing program is feasible exactly for $0\le\phi\le\bar\phi$. Set $V_I(\phi)=-\infty$ for $\phi>\bar\phi$. The uninformative zero-transfer contract is feasible for deterrence at every $\phi\ge0$, so $V_{NI}$ is finite on $[0,\infty)$.

\emph{Step 2: monotonicity and continuity.} As $\phi$ increases, the inducing feasible set shrinks and the deterring feasible set expands. Since neither objective depends on $\phi$, $V_I$ is non-increasing and $V_{NI}$ is non-decreasing. Therefore $D$ is non-increasing.

Both value functions are upper semicontinuous. Indeed, along a convergent sequence of feasible thresholds, compactness yields a convergent subsequence of optimizers attaining the limsup of the values. Continuity of the constraints makes the limiting contract feasible at the limiting threshold. Continuity of the objective then bounds the limsup by the value at that threshold.

Mixing contracts also shows that both value functions are concave on their finite domains: the reporting, participation, and investment constraints are affine, as are the objectives. They are therefore continuous in the interiors of those domains. At a finite endpoint $a$, concavity gives
\[
  F((1-t)a+tb)\ge(1-t)F(a)+tF(b),
  \qquad 0<t<1,
\]
for any other point $b$ in the domain, where $F$ denotes either value function. Letting $t\downarrow0$ and using upper semicontinuity proves continuity at the endpoint.

Consequently, $V_I$ is continuous on $[0,\bar\phi]$ and $V_{NI}$ on $[0,\infty)$. The extension of $V_I$ by $-\infty$ is upper semicontinuous, so $D$ is upper semicontinuous on $[0,\infty)$.

\emph{Step 3: the cutoff.} Reporting IC implies $U_{\thH}-U_{\thL}\ge0$, so the inducing constraint at $\phi=0$ excludes no screening contract. Hence $V_I(0)=V_\beta$. The deterrence program adds a constraint to the screening benchmark under $\alpha$, giving $V_{NI}(0)\le V_\alpha$. By Lemma~\ref{lem:Vf},
\[
  D(0)\ge V_\beta-V_\alpha\ge0.
\]

For $\phi>\bar\phi$, investment is infeasible and $D(\phi)=-\infty$.

It follows that
\[
  \Phi_D:=\{\phi\ge0:D(\phi)\ge0\}
\]
is nonempty, closed, and contained in $[0,\bar\phi]$. Monotonicity implies that $\phi\in\Phi_D$ entails $[0,\phi]\subseteq\Phi_D$. Thus $\Phi_D=[0,\phi_I]$ for some $\phi_I\in[0,\bar\phi]$, proving the result.
\end{proof}

\subsection{Proof of Proposition~\ref{prop:under}}\label{app:under}
\begin{proof}
We show that $V_I(\phi)<V_{NI}(\phi)$ whenever $\phi>\phi^{FB}$.

For any contract inducing investment, participation and the investment constraint imply
\[
  \beta U_{\thH}+(1-\beta)U_{\thL}
  =U_{\thL}+\beta(U_{\thH}-U_{\thL})
  \ge\frac{\beta}{\beta-\alpha}\phi.
\]

Since the surplus generated by each type is at most its first-best value, the principal's payoff satisfies
\[
\begin{aligned}
  W_\beta(\chi,T)
  &\le \beta K_{\theta_H}+(1-\beta)K_{\theta_L}
       -\bigl[\beta U_{\thH}+(1-\beta)U_{\thL}\bigr]\\
  &\le \beta K_{\theta_H}+(1-\beta)K_{\theta_L}
       -\frac{\beta}{\beta-\alpha}\phi
  =:L(\phi).
\end{aligned}
\]

Thus $V_I(\phi)\le L(\phi)$, with the inequality also holding when the inducing program is infeasible and $V_I(\phi)=-\infty$. The function $L$ is strictly decreasing and satisfies
\[
  L(\phi^{FB})=K_{\theta_L},
  \qquad
  \phi^{FB}=(\beta-\alpha)
             (K_{\theta_H}-K_{\theta_L}).
\]

Now offer both types $\chi^{FB}(\thL)$ at the common payment $\thL\CFBL$. This pooling contract satisfies reporting IC and participation, gives rents $U_{\thL}=0$ and $U_{\thH}=(\thL-\thH)\CFBL$, and yields the principal $K_{\theta_L}$. Optimality in the high type's first-best problem gives
\[
  K_{\theta_H}
  \ge V(\chi^{FB}(\thL))-\thH\CFBL
  =K_{\theta_L}+(\thL-\thH)\CFBL.
\]

Its investment incentive therefore satisfies
\[
  (\beta-\alpha)(U_{\thH}-U_{\thL})
  =(\beta-\alpha)(\thL-\thH)\CFBL
  \le\phi^{FB}.
\]

Consequently, this contract strictly deters investment whenever $\phi>\phi^{FB}$, giving $V_{NI}(\phi)\ge K_{\theta_L}$. Combining the two bounds, for every $\phi>\phi^{FB}$,
\[
  V_I(\phi)\le L(\phi)
  <K_{\theta_L}\le V_{NI}(\phi).
\]

The cutoff property in Proposition~\ref{prop:cutoff} therefore implies $\phi_I\le\phi^{FB}$.

The inequality is strict in the Example: Observation~\ref{obs:invest} gives $\eta_I\approx0.555<0.598\approx\eta_{FB}$. Hence efficient investment need not be induced.
\end{proof}

\subsection{Proof of Lemma~\ref{lem:etaIlb}}
\label{app:etaIlb}

\begin{proof}
Select an optimal screening contract $(\chi^\alpha,T^\alpha)$ under the distribution $\alpha$. Its high-type experiment is first-best, its low-type experiment maximizes $V-\muA C$ and has cost $\CaL$, and its rents are $U_{\thL}=0$ and $U_{\thH}=(\thL-\thH)\CaL$.

The deterrence program adds an investment constraint to the screening problem, so $V_{NI}(\phi)\le V_\alpha$. For $\eta\ge\CaL$, the selected screening contract satisfies that constraint and attains this bound. Hence
\begin{equation}
  V_{NI}(\phi)=V_\alpha
  \qquad\text{whenever }\eta\ge\CaL.
  \label{eq:VNIeq}
\end{equation}

Now let $\phi_{\alpha L}:=(\beta-\alpha)(\thL-\thH)\CaL$. At this investment cost, the same contract satisfies the investment-inducing constraint with equality. Moreover, \eqref{eq:netgain}, evaluated at $f=\alpha$ and $f'=\beta$, gives
\[
  W_\beta(\chi^\alpha,T^\alpha)
  \ge W_\alpha(\chi^\alpha,T^\alpha)
  =V_\alpha.
\]

It follows that
\[
  V_I(\phi_{\alpha L})
  \ge W_\beta(\chi^\alpha,T^\alpha)
  \ge V_\alpha
  =V_{NI}(\phi_{\alpha L}).
\]

By Proposition~\ref{prop:cutoff}, $\phi_{\alpha L}\le\phi_I$ (equivalently, $\CaL\le\eta_I$). Finally, \eqref{eq:VNIeq} implies $V_{NI}(\phi_I)=V_\alpha$.
\end{proof}

\subsection{Proof of Proposition~\ref{prop:wedge}}\label{app:wedge}
\begin{proof}
Recall that
\begin{equation}
  V_\alpha
  =\alpha K_{\theta_H}+(1-\alpha)K_{\mu_\alpha}.
  \label{eq:Valpha}
\end{equation}

Pooling both types at $\chi^{FB}(\thL)$ with the common payment $\thL\CFBL$ is feasible for the screening benchmark and yields $K_{\theta_L}$. Hence $V_\alpha\ge K_{\theta_L}$.

Define the low type's capped value
\[
  \bar K(\eta):=
  \max_{\tau\in\X:\,C(\tau)\le\eta}
  \{V(\tau)-\thL C(\tau)\}.
\]

Evaluating each type's first-best experiment in the other type's objective gives
\[
  \CFBL\le
  \eta_{FB}:=\frac{K_{\theta_H}-K_{\theta_L}}
                   {\thL-\thH}
  \le\CFBH.
\]

Together with Lemma~\ref{lem:etaIlb} and Proposition~\ref{prop:under}, this implies $\CaL\le\eta_I\le\eta_{FB}\le\CFBH$.

\emph{Step 1: the relevant payoff comparison.} For $\CbL\le\eta\le\CFBH$, Case~1 of Lemma~\ref{lem:PI} gives
\begin{equation}
  V_I(\phi)
  =\beta K_{\theta_H}+(1-\beta)\bar K(\eta)
   -\frac{\beta}{\beta-\alpha}\phi.
  \label{eq:VIcase1}
\end{equation}
Define
\[
  L(\phi):=
  \beta K_{\theta_H}+(1-\beta)K_{\theta_L}
  -\frac{\beta}{\beta-\alpha}\phi.
\]

If $\eta\ge\CFBL$, a low-type first-best experiment is feasible, so $\bar K(\eta)=K_{\theta_L}$. Consequently,
\begin{equation}
  V_I(\phi)=L(\phi)
  \qquad\text{for }\CFBL\le\eta\le\CFBH.
  \label{eq:VIeq}
\end{equation}

If $\eta<\CFBL$, every low-type first-best experiment is excluded, since all have cost $\CFBL$. Attainment of the capped maximum therefore gives $\bar K(\eta)<K_{\theta_L}$. Thus $V_I(\phi)<L(\phi)$ whenever $\CbL\le\eta<\CFBL$.

Moreover, \eqref{eq:VNIeq} gives $V_{NI}(\phi)=V_\alpha$ whenever $\eta\ge\CaL$, including at the cutoff.

\emph{Step 2: locate the cutoff relative to $\CFBL$.} Let $\phi_{FBL}:=(\beta-\alpha)(\thL-\thH)\CFBL$. Since $\CFBL\ge\CaL\ge\CbL$,
\[
  D(\phi_{FBL})
  =\beta K_{\theta_H}+(1-\beta)K_{\theta_L}
   -\beta(\thL-\thH)\CFBL-V_\alpha.
\]

The cutoff property therefore gives
\[
  \eqref{eq:regimeR}
  \quad\Longleftrightarrow\quad
  D(\phi_{FBL})\ge0
  \quad\Longleftrightarrow\quad
  \eta_I\ge\CFBL.
\]

\emph{Step 3: compute the wedge.} Let
\[
  \phi^*
  :=\frac{\beta-\alpha}{\beta}
    \bigl[\beta K_{\theta_H}
          +(1-\beta)K_{\theta_L}-V_\alpha\bigr],
\]
the unique root of $L(\phi)-V_\alpha$.

Suppose first that \eqref{eq:regimeR} holds. Then $\phi^*\ge\phi_{FBL}$, while $V_\alpha\ge K_{\theta_L}$ implies $\phi^*\le\phi^{FB}$. On $[\phi_{FBL},\phi^{FB}]$, \eqref{eq:VIeq} and \eqref{eq:VNIeq} give $D(\phi)=L(\phi)-V_\alpha$. Thus $D(\phi^*)=0$ and $D(\phi)<0$ for $\phi^*<\phi\le\phi^{FB}$. Together with $\phi_I\le\phi^{FB}$, the cutoff property implies $\phi_I=\phi^*$. Using $\phi^{FB}=(\beta-\alpha)(K_{\theta_H}-K_{\theta_L})$, we obtain
\[
  \phi^{FB}-\phi_I
  =\frac{\beta-\alpha}{\beta}
     (V_\alpha-K_{\theta_L}).
\]

Suppose instead that \eqref{eq:regimeR} fails. Then $\CaL\le\eta_I<\CFBL$. In particular, $\phi_I>0$ and $\eta_I<\CFBL\le\bar C$, so the cutoff lies in the interior of the inducing program's feasible domain. Continuity of $D$ there, together with $D(\phi_I)\ge0$ and $D(\phi)<0$ for $\phi>\phi_I$, implies $D(\phi_I)=0$. Hence
\[
  V_\alpha=V_I(\phi_I)<L(\phi_I).
\]

Since $L$ is strictly decreasing and $L(\phi^*)=V_\alpha$, we have $\phi_I<\phi^*$. Therefore,
\[
  \phi^{FB}-\phi_I
  >\phi^{FB}-\phi^*
  =\frac{\beta-\alpha}{\beta}
     (V_\alpha-K_{\theta_L}).
\]

Finally, $\CFBL\le\eta_{FB}\le\CFBH$ gives $V_I(\phi^{FB})=L(\phi^{FB})=K_{\theta_L}$, while $V_{NI}(\phi^{FB})=V_\alpha$. Thus $D(\phi^{FB})=K_{\theta_L}-V_\alpha$. Since $V_\alpha\ge K_{\theta_L}$ and $\phi_I\le\phi^{FB}$, the cutoff property implies $\phi_I=\phi^{FB}$ if and only if $V_\alpha=K_{\theta_L}$.
\end{proof}

\subsection{Proof of Proposition~\ref{prop:cutoffcs}}\label{app:cutoffcs}
\begin{proof}
Recall the low type's capped value
\[
  \bar K(\eta):=
  \max_{\tau\in\X:\,C(\tau)\le\eta}
  \{V(\tau)-\thL C(\tau)\}.
\]

\emph{Step 1: characterize the scaled cutoff.} Lemma~\ref{lem:etaIlb} and Proposition~\ref{prop:under} give
\[
  \CaL\le\eta_I\le\eta_{FB}\le\CFBH.
\]

Since $V_{NI}(\eta_I)=V_\alpha$ and investment is weakly preferred at the cutoff,
\begin{equation}
  V_I(\eta_I)\ge V_\alpha.
  \label{eq:M1}
\end{equation}

Monotonicity of $V_I$ therefore gives $V_I(\eta)\ge V_\alpha$ for $\eta\le\eta_I$. For $\eta>\eta_I\ge\CaL$, we have $V_{NI}(\eta)=V_\alpha$, and the cutoff property gives $V_I(\eta)<V_\alpha$. Hence
\begin{equation}
  \{\eta\ge0:V_I(\eta)\ge V_\alpha\}
  =[0,\eta_I].
  \label{eq:M3}
\end{equation}

\emph{Step 2: increasing $\beta$ weakly raises $\eta_I$.} Fix $\alpha$ and consider $\beta'>\beta$. Both cutoffs lie in $[\CaL,\eta_{FB}]$, and $C_L^\beta,C_L^{\beta'}\le\CaL$. Thus Lemma~\ref{lem:PI} gives the Case-1 expression with an undistorted high type throughout this common interval:
\begin{equation}
  V_I(\eta;\beta)
  =\beta K_{\theta_H}+(1-\beta)\bar K(\eta)
   -\beta(\thL-\thH)\eta.
  \label{eq:M2}
\end{equation}

Since $\bar K(\eta)\le K_{\theta_L}$ and $\eta\le\eta_{FB} =(K_{\theta_H}-K_{\theta_L})/(\thL-\thH)$,
\begin{equation}
\begin{aligned}
  K_{\theta_H}-\bar K(\eta)-(\thL-\thH)\eta
  &\ge K_{\theta_H}-K_{\theta_L}
       -(\thL-\thH)\eta\\
  &=(\thL-\thH)(\eta_{FB}-\eta)
  \ge0.
\end{aligned}
\label{eq:M4}
\end{equation}

Comparing the affine expressions in \eqref{eq:M2} therefore gives $V_I(\eta;\beta')\ge V_I(\eta;\beta)$ on the common interval. In particular, at $\eta=\eta_I(\beta)$,
\[
  V_I(\eta_I(\beta);\beta')
  \ge V_I(\eta_I(\beta);\beta)
  \ge V_\alpha.
\]

By \eqref{eq:M3}, $\eta_I(\beta')\ge\eta_I(\beta)$.

\emph{Step 3: increasing $\alpha$ weakly lowers $\eta_I$.} Fix $\beta$ and consider $\alpha'>\alpha$. At a fixed scaled cost $\eta$, the inducing program is independent of $\alpha$: its objective uses $\beta$, and its investment constraint is $U_{\thH}-U_{\thL}\ge(\thL-\thH)\eta$. Meanwhile, Lemma~\ref{lem:Vf} gives $V_{\alpha'}\ge V_\alpha$. Thus
\[
  \{\eta:V_I(\eta)\ge V_{\alpha'}\}
  \subseteq
  \{\eta:V_I(\eta)\ge V_\alpha\}.
\]

Applying \eqref{eq:M3} to both parameter values yields $\eta_I(\alpha')\le\eta_I(\alpha)$.

Finally,
\[
  \phi_I=(\beta-\alpha)(\thL-\thH)\eta_I,
  \qquad \eta_I\ge\CaL>0.
\]

Increasing $\beta$ weakly raises $\eta_I$ and strictly raises the positive factor $\beta-\alpha$. Increasing $\alpha$ weakly lowers $\eta_I$ and strictly lowers that factor. Hence $\phi_I$ is strictly increasing in $\beta$ and strictly decreasing in $\alpha$.
\end{proof}

\subsection{Proof of Claim~\ref{claim:example}}\label{app:example}
Using Lemmas~\ref{lem:PNI}--\ref{lem:PI}, the simplified programs are solved by the following characterization of \citet{azrieli2021constrained} for this environment.
\begin{lemma}[Azrieli, 2021a]\label{lem:azrieli}
\mbox{}
\begin{enumerate}[label=\arabic*.]
  \item Unconstrained: $\max_\tau\{V(\tau)-\theta C(\tau)\}$ has, for $\theta>\theta^{*}$, $\tau(m)=1$, $p_m=\tfrac12$; for $\theta\in(0,\theta^{*})$, $\tau(l)=\tau(r)=\tfrac12$, $p_l=1-p_r=[1+\exp(3/\theta)]^{-1}$; for $\theta=\theta^{*}$, any mixture.
  \item Constrained: $\max_\tau V(\tau)$ subject to $C(\tau)=\eta$ has, for $\eta\in(0,H(\tfrac12)-H(p^{*}))$, the three-posterior solution $\tau(l)=\tau(r)=\frac{\eta}{2[H(1/2)-H(p^{*})]}$, $\tau(m)=1-\frac{\eta}{H(1/2)-H(p^{*})}$, $p_l=1-p_r=p^{*}$, $p_m=\tfrac12$; for $\eta\in[H(\tfrac12)-H(p^{*}),H(\tfrac12)]$, $\tau(l)=\tau(r)=\tfrac12$ with $p_l=1-p_r$ set by $H(\tfrac12)-H(p_l)=\eta$.
\end{enumerate}
\end{lemma}

The three-posterior solution in part~2 is unique. At $\theta^*$, $v-\theta^*c$ attains its maximum of zero exactly at $p^*,1/2,1-p^*$. For $0<\eta<\eta_C$, the displayed experiment attains the bound $V(\tau)\le\theta^*\eta$. Every optimizer must therefore be supported on these three points. Bayes plausibility and the cost constraint give equal outer masses $\eta/(2\eta_C)$ and middle mass $1-\eta/\eta_C$, all strictly positive.

Applying Lemma~\ref{lem:azrieli} to the multipliers $\thH,\thL,\muA,\muB$ and to the cost constraint from Lemmas~\ref{lem:PNI}--\ref{lem:PI} yields the four regimes of Claim~\ref{claim:example}, as follows. By Proposition~\ref{prop:cutoff} investment is induced if and only if $\eta\le\eta_I$, and $\eta_I\in(\CFBL,\CFBH)$ by the constants of Section \ref{sec:example}: the low type solves \eqref{eq:PI} for $\eta\le\eta_I$ and \eqref{eq:PNI} for $\eta>\eta_I$. The high type always attains its first best: under \eqref{eq:PNI} at every $\eta$, and under \eqref{eq:PI} because it is used only for $\eta\leq \eta_I<\CFBH$ (Proposition~\ref{prop:dist}). Within \eqref{eq:PI}, $\CbL=0$, so the low type solves $\max_{C(\tau)\le\eta}\{V-\thL C\}$; by Lemma~\ref{lem:azrieli}(2) this is the three-posterior experiment for $0<\eta<\eta_C$ (the multi-posterior characterization, Appendix~\ref{app:multipost}) and the two-posterior experiment of cost $\eta$ for $\eta\in[\eta_C,\CFBL)$, reaching the $\thL$-first best $p_l=1-p_r=[1+e^{3/\thL}]^{-1}$ once $\eta\ge\CFBL$. For $\eta>\eta_I$, \eqref{eq:PNI} gives the two-posterior $\muA$-solution $p_l=1-p_r=[1+e^{3/\muA}]^{-1}$ (Lemma~\ref{lem:PNI}). All constants and regime boundaries are reported to three decimal places: $\CbL=0$, $\CaL\approx0.437$, $\CFBL\approx0.502$, $\CFBH\approx0.676$, $p^{*}\approx0.191$, $\eta_C\approx0.206$, $\eta_I\approx0.555$, $\eta_{FB}\approx0.598$.

\subsection{Derivation: lumpy information}\label{app:lumpy}
We record the \emph{lumpy-information} property of the Example. With $\eta_C=H(\tfrac12)-H(p^{*})$:
\begin{enumerate}[label=(\alph*)]
  \item for every multiplier $\theta>0$ with $\theta\neq\theta^{*}$, the unconstrained optimum $\chi_\theta\in\argmax_{\tau\in\X}\{V(\tau)-\theta C(\tau)\}$ has cost $C(\chi_\theta)\in\{0\}\cup[\eta_C,H(\tfrac12))$. In particular, the cost never lies in the open interval $(0,\eta_C)$.
  \item The cost-constrained problem $\max_{\tau:\,C(\tau)=\eta}V(\tau)$ requires three posteriors exactly when $0<\eta<\eta_C$.
\end{enumerate}
\begin{proof}
(a) By Lemma~\ref{lem:azrieli}(1): for $\theta>\theta^{*}$, $\chi_\theta$ is uninformative, so $C(\chi_\theta)=0$. For $\theta\in(0,\theta^{*})$, $\chi_\theta$ is the symmetric two-point experiment with $p_\theta=[1+\exp(3/\theta)]^{-1}$; since $\theta<\theta^{*}$ gives $3/\theta>3/\theta^{*}$, we have $p_\theta<p^{*}<\tfrac12$, and as $H$ is strictly increasing on $[0,\tfrac12]$,
\[
  C(\chi_\theta)=H(\tfrac12)-H(p_\theta)>H(\tfrac12)-H(p^{*})=\eta_C .
\]

Since $p_\theta>0$ for $\theta>0$, also $C(\chi_\theta)<H(\tfrac12)$. Hence $C(\chi_\theta)\in\{0\}\cup(\eta_C,H(\tfrac12))$, never in $(0,\eta_C)$. Part~(b) is Lemma~\ref{lem:azrieli}(2).
\end{proof}

\subsection{Proof of Claim~\ref{claim:multipost}}\label{app:multipost}
\begin{proof}
Fix an admissible $\beta$. Since $\eta=\frac{\phi} {(\beta-\alpha)(\thL-\thH)}$, open intervals of $\phi$ correspond to open intervals of $\eta$.

Only $\beta$ varies across the environments considered here. In particular, $\alpha$, the cost types, and the information-design primitives remain fixed. Thus $\CaL$, $\CFBL$, and $\eta_C$ are independent of $\beta$. The Example satisfies $\CFBL\ge\CaL>\eta_C$. By Lemma~\ref{lem:etaIlb}, $\eta_I\ge\CaL>\eta_C$. Therefore every $0<\eta<\eta_C$ lies below the investment cutoff.

The standing differentiability assumption excludes $\muB=\theta^{*}$. By the lumpy-information characterization in Appendix~\ref{app:lumpy},
\[
  \CbL=0 \quad\text{if }\muB>\theta^{*},
  \qquad
  \CbL>\eta_C \quad\text{if }\muB<\theta^{*}.
\]

\emph{Sufficiency.} Suppose $\CbL=0$. For $0<\eta<\eta_C$, Case~1 of Lemma~\ref{lem:PI} assigns the low type a solution to $\max_{\tau\in\X:\,C(\tau)\le\eta} \{V(\tau)-\thL C(\tau)\}$. The cap binds because $\eta<\eta_C<\CFBL$: otherwise mixing an optimizer with a low-type first-best experiment would strictly improve the objective while respecting the cap. At the binding cost, maximizing this objective is equivalent to maximizing $V(\tau)$ subject to $C(\tau)=\eta$. Lemma~\ref{lem:azrieli}(2) implies that every optimizer requires three posteriors.

Investment is in fact strictly preferred in this interval. The deterrence value is at most $V_\alpha$. Moreover, $V_I(\eta_C)\ge V_\alpha$, since $\eta_C<\CaL\le\eta_I$, $V_I$ is non-increasing, and $V_I(\eta_I)\ge V_\alpha$. For $0\le\eta\le\eta_C$, Lemma~\ref{lem:azrieli}(2) gives $\max_{C(\tau)=\eta}V(\tau)=\theta^{*}\eta$. Consequently,
\[
  V_I(\eta)
  =\beta K_{\theta_H}
   +(1-\beta)(\theta^{*}-\muB)\eta.
\]

Because $\muB>\theta^{*}$, this expression is strictly decreasing. Hence, for $0<\eta<\eta_C$,
\[
  V_I(\eta)>V_I(\eta_C)
  \ge V_\alpha\ge V_{NI}(\eta).
\]

Thus every optimal contract induces investment and assigns three posteriors to the low type throughout this interval. In primitive units, the interval is $(0,\phi_C)$.

\emph{Necessity.} Suppose $\CbL>0$. Then $\CbL>\eta_C$. Whenever inducing investment is optimal, consider the two cases of Lemma~\ref{lem:PI}.

If $\eta<\CbL$, Case~2 assigns the low type an unconstrained $\muB$-optimizer, because its cost $\CbL$ satisfies the floor. Since $\muB<\theta^{*}$, this is a two-posterior experiment.

If $\CbL\le\eta<\CFBL$, Case~1 assigns a binding capped optimum of cost $\eta$. Here $\eta\ge\CbL>\eta_C$, so Lemma~\ref{lem:azrieli}(2) supplies an optimal two-posterior experiment.

If $\eta\ge\CFBL$, the low type can be assigned its two-posterior first-best experiment.

Finally, if $\eta<\CaL$, the deterrence cap excludes every unconstrained $\muA$-optimizer, so $V_{NI}(\eta)<V_\alpha$. Meanwhile, $\eta<\CaL\le\eta_I$ and monotonicity give $V_I(\eta)\ge V_I(\eta_I)\ge V_\alpha$. Hence deterrence can be optimal only when $\eta\ge\CaL$. The unconstrained $\muA$-optimizer is therefore feasible for the low type's cap in Lemma~\ref{lem:PNI}. It is a two-posterior experiment. At a tie between inducing and deterring investment, either branch can consequently be selected with at most two low-type posteriors.

Thus, when $\CbL>0$, no investment-cost interval requires three low-type posteriors.
\end{proof}

\subsection{Proof of Proposition~\ref{prop:mes3}}\label{app:mes3}

\begin{proof}[Proof of Proposition~\ref{prop:mes3}]

\emph{Set-up.} Identify $\Delta(\Omega)=[0,1]$ as in Lemma~\ref{lem:blackwell} via $x=p(\omega_1)$, and write $x_0:=p_0(\omega_1)\in(0,1)$. An experiment is then a mean-$x_0$ distribution $\tau$, with $C(\tau)=\mathbb{E}_\tau[c]$, where $c(x)=H(x_0)-H(x)$. The function $c$ is strictly convex, satisfies $c(x_0)=0$, $c'(0^+)=-\infty$ and $c'(1^-)=+\infty$. Since $A$ is finite, $v(x)=\max_{a\in A}\ell_a(x)$ is the upper envelope of finitely many affine maps $\ell_a$, hence convex and piecewise affine. In particular, its slope is weakly increasing, so every kink of $v$ is \emph{convex}.

By (i), $a_0$ is the \emph{unique} optimizer at $x=0$ and $a_1$ at $x=1$. Continuity therefore implies that each remains optimal on a nondegenerate closed interval: $R_0:=\{x:\ell_{a_0}(x)=v(x)\}=[0,\underline b]$ with $\underline b>0$, and $R_1:=\{x:\ell_{a_1}(x)=v(x)\}=[\bar b,1]$ with $\bar b<1$; and neither $a_0$ nor $a_1$ optimal at $x_0$ gives $\underline b<x_0<\bar b$. We call $R_0,R_1$ the \emph{extreme regions} and say $\tau$ is \emph{extreme} if $\operatorname{supp}\tau\subseteq R_0\cup R_1$. The endpoints $\underline b,\bar b$ are convex kinks of $v$.

Recall $\tilde v_\theta:=v-\theta c$ (Figure \ref{fig:concav}). By concavification, the optimal value over mean-$x_0$ experiments is $K_\theta=\operatorname{cav}\tilde v_\theta(x_0)$. For each $\theta\ge0$, let $L_\theta$ be a supporting line of $\operatorname{cav}\tilde v_\theta$ at $x_0$, so $L_{\theta}\geq \tilde v_\theta$ on $[0,1]$ and $L_{\theta}(x_0)=\operatorname{cav}\tilde v_\theta(x_0)=K_\theta$. Let $\mathcal{S}_\theta :=\{x\in[0,1]: \tilde v_\theta(x)=L_{\theta}(x)\}$ be the contact set associated with $L_\theta$. A Bayes-plausible experiment $\tau$ is optimal at multiplier $\theta$ if and only if $\operatorname{supp}\tau\subseteq\mathcal S_\theta$ \citep{kamenica2011bayesian}. Indeed, $K_\theta-\bigl[V(\tau)-\theta C(\tau)\bigr] =\mathbb E_\tau[L_\theta(x)-\tilde v_\theta(x)]$. The integrand is continuous and nonnegative, so the gap vanishes exactly when $\tau$ is supported on $\mathcal S_\theta$.

Finally, let $\Lambda:=\{\theta>0:\ \text{some }\theta\text{-optimum is extreme}\}$ and $\theta^{*}:=\sup\Lambda$. For the Example in Section~\ref{sec:example}, this is the knife-edge multiplier $\theta^{*}$.

\begin{lemma}[Decision-theoretic band]\label{lem:band}
Under premise~(i), $\theta^{*}\in(0,\infty)$ and:
\begin{enumerate}[label=\textup{(\arabic*)}]
  \item $\theta^{*}\in\Lambda$; there is a $\theta^{*}$-optimum $\tilde\tau$ supported on two points $x_L\in(0,\underline b)$, $x_R\in(\bar b,1)$ (one in each extreme region, both interior to it), and $\bar c:=C(\tilde\tau)=\max\{C(\tau):\tau\text{ a }\theta^{*}\text{-optimum}\}$, with $0<\bar c<H(x_0)$;
  \item there is a $\theta^{*}$-optimum $\hat\tau$ with $C(\hat\tau)<\bar c$, and every value in $[C(\hat\tau),\bar c]$ is the cost of some $\theta^{*}$-optimum;
  \item there is $\underline c$ with $C(\hat\tau)\le\underline c<\bar c$ such that for every $c\in(\underline c,\bar c)$, every solution of $\max_{\tau\in\X:\,C(\tau)\le c}V(\tau)$ has at least three posteriors.

    Here, $0\le\underline c<\bar c$, with $\underline c>0$ whenever $\dirac_{x_0}$ is not a $\theta^{*}$-optimum.
\end{enumerate}
\end{lemma}

\begin{proof}
Let $M:=K_0=\max_{\tau\in\X}V(\tau)$ be the maximal value, attained by full information $\{0,1\}$ (of cost $C(\{0,1\})=H(x_0)$ since $H(0)=H(1)=0$); as $v$ is convex, $M$ equals the value at $x_0$ of the chord of $v$ over $[0,1]$ (the $\theta=0$ optimum, at the endpoints). For $\theta>0$, by contrast, no optimal posterior sits at an endpoint: there $-\theta c'$ dominates the finite $v'$, so $c'(0^+)=-\infty$, $c'(1^-)=+\infty$ give $\tilde v_\theta=v-\theta c$ slope $+\infty$ at $0^+$ and $-\infty$ at $1^-$. Therefore, $\tilde v_\theta(0)<L_{\theta}(0)$ and $\tilde v_\theta(1)<L_{\theta}(1)$. We have $0,1\notin \mathcal{S}_\theta$, and every $\theta$-optimal posterior lies in $(0,1)$.

\emph{Step 1: the widest contact pair maximizes cost.} Fix $\theta>0$. Let $l_\theta:=\min\mathcal S_\theta$ and $r_\theta:=\max\mathcal S_\theta$. Every $\theta$-optimizer is supported on $\mathcal S_\theta$ and has mean $x_0$, so $l_\theta\le x_0\le r_\theta$.

If $l_\theta=r_\theta$, both equal $x_0$, and the only optimizer is $\delta_{x_0}$. Otherwise, let $\tilde\tau(\theta)$ be the mean-$x_0$ experiment on $\{l_\theta,r_\theta\}$. Its support consists of contact points, so it is $\theta$-optimal. For every $x\in[l_\theta,r_\theta]$, convexity gives
\[
  c(x)\le
  \frac{r_\theta-x}{r_\theta-l_\theta}c(l_\theta)
  +\frac{x-l_\theta}{r_\theta-l_\theta}c(r_\theta).
\]

Taking expectations under any $\theta$-optimizer $\tau$ and using $\mathbb E_\tau[x]=x_0$ yields
\[
  C(\tau)\le
  \frac{r_\theta-x_0}{r_\theta-l_\theta}c(l_\theta)
  +\frac{x_0-l_\theta}{r_\theta-l_\theta}c(r_\theta)
  =C(\tilde\tau(\theta)).
\]

Strict convexity makes equality possible only when $\tau$ is supported on the endpoints. Bayes plausibility then uniquely determines their probabilities. Thus $\tilde\tau(\theta)$ is the unique cost-maximizing $\theta$-optimizer. It has two posteriors when $l_\theta<x_0<r_\theta$ and reduces to $\delta_{x_0}$ when the prior equals an endpoint.

\emph{Step 2: $\Lambda$ bounded.} Let $\pi$ be the mean-$x_0$ two-point experiment on $\{\underline b,\bar b\}$ (the narrowest extreme spread) and $\kappa:=C(\pi)>0$; here $\kappa>0$ because $\pi$ is informative and $c$ is strictly convex with $c(x_0)=0$, so by Jensen $C(\pi)=\mathbb E_\pi[c]>c(x_0)=0$. If $\theta\in\Lambda$, an extreme $\theta$-optimum $\tau$ is supported on $R_0\cup R_1=[0,\underline b]\cup[\bar b,1]$. Let $\ell$ be the affine function agreeing with $c$ at $\underline b$ and $\bar b$. Convexity gives $c\ge\ell$ on $R_0\cup R_1$. Since $\tau$ and $\pi$ both have mean $x_0$,
\[
  C(\tau)\ge\mathbb E_\tau[\ell]
  =\ell(x_0)=C(\pi)=\kappa.
\] 

Optimality against $\dirac_{x_0}$ gives $V(\tau)-\theta C(\tau)\ge v(x_0)$, whence $\theta\kappa\le\theta C(\tau)\le V(\tau)-v(x_0)\le M-v(x_0)$. Here $M>v(x_0)$: a convex $v$ has $v(x_0)\le M$, with equality only if $v$ is affine on $[0,1]$ (a convex function meeting its chord at an interior point coincides with it), which is excluded by the kink of $v$ at $\bar b$. Hence $\theta\le(M-v(x_0))/\kappa<\infty$, so $\Lambda$ is bounded above.

\emph{Step 3: $\Lambda\supseteq(0,\underline\theta)$, nonempty.} Suppose $r_\theta<\bar b$. Every $\theta$-optimum is then supported in $[l_\theta,r_\theta]\subseteq[0,\bar b)$, so its value is at most the chord of $v$ over $[0,\bar b]$ at $x_0$, call it $M_L$; and $M_L<M$, because extending the chord $[0,\bar b]$ linearly to $x=1$ underestimates $v(1)$ (the action $a_1$, optimal on $[\bar b,1]$, is strictly steeper than that chord), so the longer chord $[0,1]$ lies strictly higher at the interior point $x_0$. Thus $K_{\theta}\le M_L$; but comparing with $\{0,1\}$, $K_{\theta}\ge M-\theta H(x_0)$, so $r_\theta<\bar b$ forces $\theta\ge(M-M_L)/H(x_0)>0$. Symmetrically, with $M_R<M$ the chord of $v$ over $[\underline b,1]$ at $x_0$, $l_\theta>\underline b$ forces $\theta\ge(M-M_R)/H(x_0)>0$. Hence for $\theta<\underline\theta:=\min\{M-M_L,\,M-M_R\}/H(x_0)$ we have $l_\theta\le\underline b$ and $r_\theta\ge\bar b$; with $0<l_\theta\le\underline b<x_0<\bar b\le r_\theta<1$, the cost-maximal optimum $\tilde\tau(\theta)$ on $\{l_\theta,r_\theta\}$ has one posterior in $R_0$ and one in $R_1$, i.e.\ it is extreme. So $(0,\underline\theta)\subseteq\Lambda$, and $\underline\theta\le\theta^{*}<\infty$, i.e.\ $\theta^{*}\in(0,\infty)$.

\emph{Step 4: $\theta^{*}\in\Lambda$ and the extreme optimum.} By Steps~2--3, $\Lambda$ is nonempty and $\theta^{*}=\sup\Lambda\in(0,\infty)$. Choose $\theta_m\in\Lambda$ with $\theta_m\to\theta^{*}$, and select an extreme $\theta_m$-optimizer $\tau_m$ for each $m$. Compactness of $\X$ gives a subsequence converging weakly to some $\tau\in\X$.

For every $\sigma\in\X$, optimality gives
\[
  V(\tau_m)-\theta_m C(\tau_m)
  \ge V(\sigma)-\theta_m C(\sigma).
\]

Passing to the limit, using continuity of $V$ and $C$, shows that $\tau$ is a $\theta^{*}$-optimizer. Moreover, $R_0\cup R_1$ is closed and $\tau_m(R_0\cup R_1)=1$, so $\tau(R_0\cup R_1)=1$. Thus $\tau$ is extreme. Its mean is $x_0\in(\underline b,\bar b)$, so it places positive probability in both extreme regions. Hence $\theta^{*}\in\Lambda$. Let $\tilde\tau:=\tilde\tau(\theta^{*})$ on $\{x_L,x_R\}$, $x_L=l_{\theta^{*}}$, $x_R=r_{\theta^{*}}$. Because $\theta^{*}\in \Lambda$, some $\theta^{*}$-optimum is extreme; being supported in $[x_L,x_R]=[\min \mathcal{S}_{\theta^{*}},\max \mathcal{S}_{\theta^{*}}]$ with a posterior in each of $R_0,R_1$, it forces $x_L\le\underline b$ and $x_R\ge\bar b$. Moreover $x_L\ne\underline b$: the convex kink $\underline b$ (left slope $<$ right slope) cannot lie in the contact set $\mathcal{S}_{\theta^{*}}$, since at a contact point the nonnegative $\operatorname{cav}\tilde v_{\theta^{*}}-\tilde v_{\theta^{*}}$ attains its minimum $0$, forcing $(\operatorname{cav}\tilde v_{\theta^{*}})'_-\le(\tilde v_{\theta^{*}})'_-<(\tilde v_{\theta^{*}})'_+\le(\operatorname{cav}\tilde v_{\theta^{*}})'_+$ there, which contradicts concavity of $\operatorname{cav}\tilde v_{\theta^{*}}$. With $x_L>0$ likewise, $x_L\in (0,\underline b)$; symmetrically $x_R\in(\bar b,1)$. Since $\tilde\tau$ is a two-point optimum spanning $x_0$, $\operatorname{cav}\tilde v_{\theta^{*}}$ is affine on $[x_L,x_R]$, equal there to to the supporting line $L^{*}:=L_{\theta^{*}}$. As $x_{L}=\min \mathcal{S}_{\theta^{*}}$ and $x_{R}=\max \mathcal{S}_{\theta^{*}}$, $\mathcal{S}_{\theta^{*}}=\{x\in [x_{L},x_{R}]: \tilde v_{\theta^{*}}(x)=L^{*}(x)\}$, which is \emph{finite}: on each affine piece of $v$, $\tilde v_{\theta^{*}}-L^{*}$ is strictly concave and $\le0$, hence vanishes at most once, so $|\mathcal{S}_{\theta^{*}}|\le|A|$. Every $\theta^{*}$-optimum is supported in $\mathcal{S}_{\theta^{*}}\subseteq[x_L,x_R]$ with mean $x_0$, hence $\preceq_{cx}\tilde\tau$ in the convex order and $C(\cdot)\le C(\tilde\tau)$, with equality only at $\tilde\tau$ (a strict spread otherwise). Thus $\bar c:=C(\tilde\tau)$ is the unique maximal $\theta^{*}$-optimal cost, and $0<\bar c<H(x_0)$.

\emph{Step 5: coexistence $\hat\tau$.} Take $\theta^{(m)}\downarrow\theta^{*}$. As $\theta^{(m)}>\theta^{*}=\sup\Lambda$, no $\theta^{(m)}$-optimum is extreme; in particular the cost-maximal two-point optimum $\tilde\tau(\theta^{(m)})$ on $\{l_{\theta^{(m)}},r_{\theta^{(m)}}\}$ is not extreme. The map $\theta\mapsto C(\tilde\tau(\theta))$ is non-increasing (Lemma~\ref{lem:mono}), so $\lim_m C(\tilde\tau(\theta^{(m)}))$ exists; pass to a subsequence along which $\tilde\tau(\theta^{(m)})$ converges weakly to a limit $\hat\tau$, a $\theta^{*}$-optimum, whose cost is $C(\hat\tau)=\lim_m C(\tilde\tau(\theta^{(m)}))\le \bar c$ by continuity of $C$. Moreover, $\hat\tau$ minimizes cost among all $\theta^{*}$-optimizers. To see this, let $\sigma$ be any $\theta^{*}$-optimizer. Since $\theta^{(m)}>\theta^{*}$, Lemma~\ref{lem:mono} gives
\[
  C(\tilde\tau(\theta^{(m)}))\le C(\sigma).
\]

Taking limits yields $C(\hat\tau)\le C(\sigma)$. If $C(\hat\tau)=\bar c$ then $\hat\tau=\tilde\tau$ by the uniqueness in Step~4; as $\tilde\tau(\theta^{(m)})$ and $\tilde\tau$ are two-point measures with distinct atoms, their atoms converge, $l_{\theta^{(m)}}\to x_L\in(0,\underline b)$ and $r_{\theta^{(m)}}\to x_R\in(\bar b,1)$, so for large $m$ both atoms of $\tilde\tau(\theta^{(m)})$ lie in the \emph{open} extreme regions, making $\tilde\tau(\theta^{(m)})$ extreme, contradicting $\theta^{(m)}\notin\Lambda$. Hence $C(\hat\tau)<\bar c$. For $\mu\in[0,1]$ the mixture $\mu \tilde\tau+(1-\mu) \hat\tau$ is a $\theta^{*}$-optimum (the objective is linear, both are optimal) of cost $\mu\bar c+(1-\mu)C(\hat\tau)$, so every value in $[C(\hat\tau),\bar c]$ is a $\theta^{*}$-optimal cost.

\emph{Step 6: the band.} For $c\in[C(\hat\tau),\bar c]$ pick a $\theta^{*}$-optimal mixture $\tau_c$ of cost exactly $c$; then $V(\tau_c)-\theta^{*} c=K(\theta^{*})$, so $\max_{C(\tau)\le c}V\ge V(\tau_c)=K_{\theta^{*}}+\theta^{*} c$. Conversely any feasible $\tau$ ($C(\tau)\le c$) has $V(\tau)\le K_{\theta^{*}}+\theta^{*} C(\tau)\le K_{\theta^{*}}+\theta^{*} c$. Hence
\[
  \max_{\tau\in\X:\,C(\tau)\le c}V(\tau)=K_{\theta^{*}}+\theta^{*} c
  \qquad(c\in[C(\hat\tau),\bar c]),
\]
and any maximizer $\tau^\circ$ meets both inequalities with equality: since $\theta^{*}>0$ this forces $C(\tau^\circ)=c$ (the cap binds) and $V(\tau^\circ) -\theta^{*} C(\tau^\circ)=K_{\theta^{*}}$, i.e.\ $\tau^\circ$ is a $\theta^{*}$-optimum, supported on the finite set $\mathcal{S}_{\theta^{*}}$. Let $\mathcal C_2$ be the (finite) set of costs of $\theta^{*}$-optima supported on \emph{at most two} points of $\mathcal{S}_{\theta^{*}}$; by Step~4 its maximum is $\bar c$, attained only by $\tilde\tau$, so every other element is $<\bar c$. Put
\[
  \underline c:=\max\bigl(\{C(\hat\tau)\}\cup(\mathcal C_2\setminus\{\bar c\})\bigr)
  \ \in\ [\,C(\hat\tau),\ \bar c\,).
\]

For $c\in(\underline c,\bar c)$ the maximizer $\tau^\circ$ is a $\theta^{*}$-optimum of cost $c$, and $c\notin\mathcal C_2$ (it exceeds every element of $\mathcal C_2$ other than $\bar c$, and $c<\bar c$), so $\tau^\circ$ cannot be supported on $\le2$ points; it uses at least three posteriors. Finally $\underline c\ge C(\hat\tau)\ge0$; and if $\dirac_{x_0}$ is \emph{not} a $\theta^{*}$-optimum then $C(\hat\tau)>0$ --- the minimal-cost optimum $\hat\tau$ cannot be the (unique) zero-cost experiment $\dirac_{x_0}$ --- so $\underline c>0$.
\end{proof}

We now apply Lemma~\ref{lem:band} to the contracting problem.

\emph{Step 1: locate the band relative to the reference costs.} By decisiveness~(ii), a $\muA$-optimizer has a posterior in each extreme region. The widest pair of contact points therefore supports an extreme $\muA$-optimizer, so $\muA\in\Lambda$ and $\muA\le\theta^{*}$. Lemma~\ref{lem:band} gives $\theta^{*}$-optimizers with different costs. Since $K$ is differentiable at $\muA$, we cannot have $\muA=\theta^{*}$. Thus, $\thL<\muA<\theta^{*}$. Cost monotonicity gives $\CaL\ge\bar c$, and Lemma~\ref{lem:etaIlb} then implies $\eta_I\ge\CaL\ge\bar c$.

\emph{Step 2: investment is strictly preferred.} Fix $\eta\in(\underline c,\bar c)$. Every unconstrained $\muA$-optimizer has cost $\CaL$, which exceeds $\eta$. Hence the deterrence cap excludes every such optimizer. Since the capped maximum is attained, Lemma~\ref{lem:PNI} implies $V_{NI}(\eta)<V_\alpha$.

On the other hand, monotonicity of $V_I$ and $\eta<\bar c\le\eta_I$ give $V_I(\eta)\ge V_I(\eta_I)\ge V_\alpha$, where the last inequality follows from $V_I(\eta_I)\ge V_{NI}(\eta_I)=V_\alpha$. Therefore $V_I(\eta)>V_{NI}(\eta)$ throughout the band.

\emph{Step 3: every optimal low-type experiment requires at least three posteriors.} By silencing~(iii), $\CbL=0$. Lemma~\ref{lem:PI} therefore places the low type in Case~1, solving $\max_{\tau\in\X:\,C(\tau)\le\eta} \{V(\tau)-\thL C(\tau)\}$. For any feasible experiment,
\[
\begin{aligned}
  V(\tau)-\thL C(\tau)
  &=[V(\tau)-\theta^{*}C(\tau)]
    +(\theta^{*}-\thL)C(\tau)\\
  &\le K_{\theta^{*}}
    +(\theta^{*}-\thL)\eta.
\end{aligned}
\]

Lemma~\ref{lem:band} supplies a $\theta^{*}$-optimizer of cost exactly $\eta$, which attains this bound. Since $\theta^{*}>\thL$, every maximizer must have cost $\eta$ and be $\theta^{*}$-optimal. It therefore also attains $\max_{C(\tau)\le\eta}V(\tau)=K_{\theta^{*}}+\theta^{*}\eta$. By Lemma~\ref{lem:band}(3), every such experiment uses at least three posteriors.

Finally, define $\phi_{\mathrm{low}}:=(\beta-\alpha)(\thL-\thH)\underline c$ and $\phi_{\mathrm{high}}:=(\beta-\alpha)(\thL-\thH)\bar c$. Since $\bar c>\underline c\ge0$, these endpoints satisfy $\phi_{\mathrm{high}}>\phi_{\mathrm{low}}\ge0$. The preceding argument applies to every $\phi\in(\phi_{\mathrm{low}},\phi_{\mathrm{high}})$, proving the proposition.
\end{proof}

\begin{remark}
The knife-edge multiplier deserves some discussion. It generalizes the Example's $\theta^{*}=1/\ln\tfrac{1+\sqrt5}{2}$. In general, however, information acquisition need not cease above $\theta^{*}$: a non-extreme optimum may remain informative. This is why $\theta^{*}$ is defined by the disappearance of \emph{extreme} support, rather than by a transition from informative to uninformative experiments.
\end{remark}

\subsection{The \texorpdfstring{$n$}{n}-type framework: screening and ironing}\label{app:ntype-pf}

We first characterize implementable allocations and minimal rents, then solve the screening benchmark $V_g$ using virtual multipliers and ironing.

\begin{lemma}[Screening characterization]\label{lem:ntype-screen}
Fix $g$ of full support.
\begin{enumerate}[label=\textup{(\roman*)}]
  \item Every $(\chi,T)\in\mathcal P_n$ has a monotone allocation $C_1\le C_2\le\dots\le C_n$. \emph{(Forced by \eqref{eq:IC} alone; no regularity.)}
  \item Conversely, every allocation with $C_1\le\dots\le C_n$ is implementable, and among implementing transfers the principal's payoff $\sum_i g_i S_i$ is maximized by the minimal rents
    \[
      U_1=0,\qquad U_i-U_{i-1}=(\theta_{i-1}-\theta_i)\,C_{i-1}\quad(2\le i\le n),
      \quad\text{i.e.}\quad U_i=\sum_{k=1}^{i-1}(\theta_k-\theta_{k+1})\,C_k .
    \]
  \item Consequently,
\[
  V_g=\max_{\substack{\chi_i\in\X\\ C_1\le\cdots\le C_n}}
  \sum_{i=1}^n g_i
  [V(\chi_i)-\lambda_i C(\chi_i)].
\]

If $g$ is regular, the monotonicity constraint can be satisfied without reducing the value of the pointwise relaxation. Specifically, choose a common optimizer for types with the same virtual multiplier and choose
\[
  \chi_i\in\argmax_{\tau\in\X}
  \{V(\tau)-\lambda_i C(\tau)\}
\]
at each distinct multiplier. These choices have $C_1\le\cdots\le C_n$ and, together with the minimal rents in part~(ii), form an optimal screening contract. Moreover, $\lambda_i\ge\theta_i$ and $\lambda_n=\theta_n$.
\end{enumerate}
\end{lemma}
\begin{proof}
\emph{(i)} Adding the reporting constraints between types $i$ and $i+1$ gives
\[
  (\theta_i-\theta_{i+1})(C_{i+1}-C_i)\ge0.
\]

Since $\theta_i>\theta_{i+1}$, we obtain $C_i\le C_{i+1}$.

\emph{(ii)} Fix an allocation with $C_1\le\cdots\le C_n$. Participation and the adjacent reporting constraints imply $U_1\ge0$ and $U_{i+1}\ge U_i+(\theta_i-\theta_{i+1})C_i$.

Hence every implementing rent profile satisfies
\[
  U_i\ge U_i^*
  :=\sum_{k=1}^{i-1}
       (\theta_k-\theta_{k+1})C_k,
  \qquad U_1^*=0.
\]

These lower bounds are jointly attainable. Indeed, for every $i<j$, cost monotonicity gives
\[
  (\theta_i-\theta_j)C_i
  \le U_j^*-U_i^*
  =\sum_{k=i}^{j-1}(\theta_k-\theta_{k+1})C_k
  \le(\theta_i-\theta_j)C_j.
\]

These are the two reporting constraints between types $i$ and $j$. The rents are nonnegative, so $T_i=\theta_iC_i+U_i^*$ also satisfies participation and transfer nonnegativity.

For a fixed allocation, the principal's payoff is $\sum_i g_i[V(\chi_i)-\theta_iC_i-U_i]$. Thus the feasible profile $U^*$, which minimizes every rent, maximizes her payoff.

\emph{(iii)} Substituting the minimal rents and interchanging the order of summation gives
\[
  \sum_i g_iU_i^*
  =\sum_{k=1}^{n-1}
     G_{k+1}(\theta_k-\theta_{k+1})C_k.
\]

By the definition of the virtual multipliers,
\[
  \sum_i g_i[V(\chi_i)-\theta_iC_i-U_i^*]
  =\sum_i g_i[V(\chi_i)-\lambda_iC_i].
\]

Parts~(i)--(ii) therefore give the stated expression for $V_g$.

Under regularity, $\lambda_1\ge\cdots\ge\lambda_n$. Choose a common optimizer for each distinct multiplier. Lemma~\ref{lem:mono} orders costs across strictly different multipliers, while common selections give equal costs at tied multipliers. The resulting allocation satisfies cost monotonicity and attains the pointwise relaxed value. Together with minimal rents, it is an optimal screening contract.

Finally, $\lambda_i\ge\theta_i$ and $\lambda_n=\theta_n$ follow directly from Definition~\ref{as:regular}.
\end{proof}

For a nonregular distribution, the pointwise relaxation need not admit a cost-monotone allocation. We resolve this problem by ironing the virtual multipliers. This construction characterizes the unconstrained screening benchmark $V_g$. The additional investment constraints in $V_I$ and $V_{NI}$ require separate analysis.

\begin{lemma}[Ironing the screening problem.]\label{lem:iron}
Fix $g$ of full support, not assumed regular, and recall from Lemma~\ref{lem:ntype-screen}(iii) the identity $V_g=\max_{C_1\le\dots\le C_n}\sum_{i=1}^n g_i[V(\chi_i)-\lambda_i C(\chi_i)]$, which holds for every $g$ (regularity is used only \emph{after} it, to drop the constraint). Recall that $G_i=\sum_{k\ge i}g_k$.

Define
\[
  q_0=A_0=0,\qquad
  q_j:=\sum_{i=1}^j g_i,\qquad
  A_j:=\sum_{i=1}^j g_i\lambda_i
  \quad(1\le j\le n).
\]

Let $A(q)$ be the piecewise-linear interpolation of the points $(q_j,A_j)$, and let $\overline A$ be its least concave majorant on $[0,1]$. Define the ironed multipliers by
\[
  \overline\lambda_i
  :=\frac{\overline A(q_i)-\overline A(q_{i-1})}{g_i}.
\]

Partition the types at every contact index $j$ satisfying $\overline A(q_j)=A_j$. An ironing block $B=\{i,\ldots,j\}$ lies between consecutive contact indices $i-1$ and $j$. The majorant is affine on $[q_{i-1},q_j]$. Adjacent blocks may have the same slope.

\begin{enumerate}[label=\textup{(\alph*)}]
  \item \emph{(Tail rent identity.)} For every $i$, $\displaystyle\sum_{k=i}^{n} g_k\lambda_k=\theta_i\,G_i$.
  \item \emph{(Block-average lower bound.)} For any consecutive block $B=\{i,\ldots,j\}$, define
\[
  \bar\lambda_B
  :=\frac{\sum_{k=i}^{j}g_k\lambda_k}
           {\sum_{k=i}^{j}g_k}.
\]

If $j<n$, then
\[
  \bar\lambda_B
  =\theta_i+
    \frac{G_{j+1}(\theta_i-\theta_{j+1})}
         {G_i-G_{j+1}}
  >\theta_i\ge\theta_m
  \qquad\text{for every }m\in B.
\]
  \item \emph{(Ironed characterization.)} The ironed multipliers are non-increasing. On each ironing block $B$, the ironed multipliers $\overline\lambda_k$, $k\in B$, equal the block average $\bar\lambda_B$ defined in part~(b).

For each distinct ironed multiplier, select one optimizer of
\[
  \max_{\tau\in\X}
  \{V(\tau)-\overline\lambda_i C(\tau)\}.
\]

Assign this common optimizer to every type with that ironed multiplier. The resulting allocation is constant within each ironing block and has non-decreasing costs. Together with the minimal rents of Lemma~\ref{lem:ntype-screen}(ii), it attains $V_g$.

The most efficient type forms a singleton ironing block and has $\overline\lambda_n=\theta_n$. This does not require its assigned experiment to differ from those assigned to other blocks.

If $g$ is regular, $\overline A=A$ and $\overline\lambda_i=\lambda_i$ for every $i$. The construction then reduces to the coordinated pointwise selection in Lemma~\ref{lem:ntype-screen}(iii).
\end{enumerate}
\end{lemma}
\begin{proof}
\emph{(a)} Since $g_k=G_k-G_{k+1}$, the virtual-multiplier formula gives
\[
  g_k\lambda_k
  =\theta_kG_k-\theta_{k+1}G_{k+1}
  \qquad(k<n),
\]
while $g_n\lambda_n=\theta_nG_n$. Summing from $k=i$ to $n$ telescopes to
\[
  \sum_{k=i}^n g_k\lambda_k=\theta_iG_i.
\]

\emph{(b)} For $B=\{i,\ldots,j\}$ with $j<n$, subtracting the tail identities at $i$ and $j+1$ gives
\[
\begin{aligned}
  \bar\lambda_B
  &=\frac{\theta_iG_i-\theta_{j+1}G_{j+1}}
          {G_i-G_{j+1}}\\
  &=\theta_i+
    \frac{G_{j+1}(\theta_i-\theta_{j+1})}
         {G_i-G_{j+1}}
  >\theta_i.
\end{aligned}
\]

The strict inequality follows from full support and strictly ordered types. Since $\theta_i\ge\theta_m$ for every $m\in B$, the claimed bound follows.

\emph{(c)} Concavity of $\overline A$ makes the ironed multipliers non-increasing. On an ironing block, $\overline A$ is affine and agrees with $A$ at both endpoints. Its slope therefore equals the block average $\bar\lambda_B$.

Set $R_j:=\overline A(q_j)-A_j$. Then $R_j\ge0$ and $R_0=R_n=0$. For any allocation with non-decreasing costs, summation by parts gives
\[
\begin{aligned}
  \sum_{i=1}^n g_i
    (\lambda_i-\overline\lambda_i)C_i
  &=-\sum_{i=1}^n(R_i-R_{i-1})C_i\\
  &=\sum_{j=1}^{n-1}R_j(C_{j+1}-C_j)
  \ge0.
\end{aligned}
\]

Consequently,
\[
\begin{aligned}
  \sum_i g_i[V(\chi_i)-\lambda_iC_i]
  &\le
  \sum_i g_i[V(\chi_i)-\overline\lambda_iC_i]\\
  &\le
  \sum_i g_i
  \max_{\tau\in\X}
  \{V(\tau)-\overline\lambda_iC(\tau)\}.
\end{aligned}
\]

To attain this bound, choose a common optimizer for each distinct ironed multiplier. Lemma~\ref{lem:mono} and the coordinated choices give non-decreasing costs, and pointwise optimality makes the second inequality an equality. The first inequality is also an equality: if $R_j>0$, then $j$ lies inside an ironing block, so the construction gives $C_{j+1}=C_j$. Thus every term $R_j(C_{j+1}-C_j)$ vanishes. The resulting allocation, implemented with the minimal rents from Lemma~\ref{lem:ntype-screen}(ii), therefore attains $V_g$.

To establish the top-type claim, suppose the final ironing block were $\{i,\ldots,n\}$ with $i<n$. By part~(a), its average multiplier would be $\theta_i$. Since $i-1$ is a contact index and $\overline A$ is affine on this block,
\[
  \overline A(q_i)
  =A_{i-1}+g_i\theta_i
  <A_{i-1}+g_i\lambda_i
  =A_i,
\]
where $\lambda_i>\theta_i$ follows from full support. This contradicts $\overline A\ge A$. Hence the final block is $\{n\}$ and $\overline\lambda_n=\lambda_n=\theta_n$.

Finally, if $g$ is regular, $A$ is already concave. Thus $\overline A=A$, every index is a contact index, and ironing leaves all multipliers unchanged.
\end{proof}

\subsection{Proof of Proposition~\ref{prop:Vf-ntype} (value monotonicity)}\label{app:gen-vf-pf}

\begin{proof}

\emph{Part 1: convexity.} For each fixed contract, the principal's expected payoff is affine in $g$. Since $\mathcal P_n$ is independent of $g$, $V_g=\sup_{(\chi,T)\in\mathcal P_n} \sum_i g_i[V(\chi_i)-T_i]$ is a pointwise supremum of affine functions, hence convex along every affine path in its domain.
       
\emph{Part 2: monotonicity.} Select the optimal screening contract $(\chi,T)$ constructed by Lemma~\ref{lem:iron}(c), with coordinated experiment choices and minimal rents. Write $C_i=C(\chi_i)$ and $S_i=V(\chi_i)-T_i$.

We first show that $S_1\le\cdots\le S_n$. The construction gives $C_{i-1}\le C_i$ and
\[
  \chi_i\in\argmax_{\tau\in\X}
  \{V(\tau)-\overline\lambda_i C(\tau)\}.
\]

Moreover, $\overline\lambda_i\ge\theta_i$: Lemma~\ref{lem:iron}(b) establishes this for every block not containing $n$, while $\overline\lambda_n=\theta_n$ for the top singleton.

For $i\ge2$, minimal rents satisfy $U_i-U_{i-1} =(\theta_{i-1}-\theta_i)C_{i-1}$. Using $T_i=\theta_iC_i+U_i$ and optimality of $\chi_i$ at $\overline\lambda_i$, we obtain
\[
\begin{aligned}
  S_i-S_{i-1}
  &=V(\chi_i)-V(\chi_{i-1})
    -\theta_i(C_i-C_{i-1})
  \ge(\overline\lambda_i-\theta_i)
       (C_i-C_{i-1})
  \ge0.
\end{aligned}
\]

Since $\mathcal P_n$ does not depend on the type distribution, this contract is also feasible under $g'$. Consequently,
\[
\begin{aligned}
  V_{g'}-V_g
  &\ge\sum_{i=1}^n(g'_i-g_i)S_i
  =\sum_{k=2}^n(G'_k-G_k)(S_k-S_{k-1})
  \ge0,
\end{aligned}
\]
where the equality follows by summation by parts and the final inequality follows from FOSD and the payoff ordering established above.

For $n=2$, setting $\theta_1=\thL$, $\theta_2=\thH$, and $g=(1-f,f)$ recovers the screening value $V_f$ and the conclusions of Lemma~\ref{lem:Vf}.
\end{proof}

\subsection{Proof of Proposition~\ref{prop:cutoff-ntype} (the investment cutoff)}\label{app:gen-cutoff-pf}

\begin{proof}
For a contract $z=(\chi,T)\in\mathcal P_n$, write
\[
  B(z):=\langle g'-g,U\rangle,\qquad
  J_I(z):=\sum_i g'_iS_i,\qquad
  J_{NI}(z):=\sum_i g_iS_i.
\]

Reporting incentive compatibility implies
\begin{equation}
  (\theta_{k-1}-\theta_k)C_{k-1}
  \le U_k-U_{k-1}
  \le(\theta_{k-1}-\theta_k)C_k,
  \qquad k=2,\ldots,n.
  \tag{$\ddagger$}\label{eq:gap-bound}
\end{equation}

\emph{Part 1: the investment region is a closed interval whenever it is nonempty.} The rent-gap bounds imply $U_i\ge U_1\ge0$. Subtracting $U_1$ from every transfer preserves reporting incentives, participation, and the investment wedge, while increasing the principal's payoff by $U_1$. Thus both programs may be restricted to
\[
  \mathcal Z
  :=\{(\chi,T)\in\mathcal P_n:U_1=0\}.
\]

On this domain, we have $0\le U_i\le(\theta_1-\theta_i)\bar C$, and $0\le T_i\le\theta_1\bar C$. Since $\X$ is compact and convex and $V,C$ are continuous and affine under the maintained assumptions, $\mathcal Z$ is compact and convex. The functions $B,J_I,J_{NI}$ are also continuous and affine. Consequently, every nonempty program attains its maximum.

As $\phi$ increases, the inducing feasible set shrinks and the deterring feasible set expands. Hence $V_I$ is non-increasing, $V_{NI}$ is non-decreasing, and $D=V_I-V_{NI}$ is non-increasing. Moreover, \eqref{eq:abel-wedge}, \eqref{eq:gap-bound}, and $|G'_k-G_k|\le1$ give
\[
  |B(z)|
  \le\sum_{k=2}^n(\theta_{k-1}-\theta_k)C_k
  \le(\theta_1-\theta_n)\bar C
  =:\bar\phi.
\]

Let $\phi_{\max}:=\max_{z\in\mathcal Z}B(z)$. The uninformative zero-transfer contract has wedge zero, so $0\le\phi_{\max}\le\bar\phi$. The inducing program is feasible exactly on $[0,\phi_{\max}]$. Set $V_I(\phi)=-\infty$ outside this interval. The same zero-transfer contract is feasible for deterrence at every $\phi\ge0$, so $V_{NI}$ is finite on $[0,\infty)$.

Both value functions are upper semicontinuous. Indeed, along any convergent sequence of feasible thresholds, compactness yields a convergent subsequence of optimizers attaining the limsup of the values. Continuity of the constraints makes the limit contract feasible at the limiting threshold, and continuity of the objective bounds that limsup by the value there. The extension of $V_I$ by $-\infty$ is therefore also upper semicontinuous.

Mixing contracts shows that $V_I$ and $V_{NI}$ are concave on their finite domains. They are consequently continuous on the interiors of those domains. Continuity also holds at each finite endpoint: for an endpoint $a$ and another point $b$ in the domain, concavity gives
\[
  F((1-t)a+tb)\ge(1-t)F(a)+tF(b),
  \qquad 0<t<1,
\]
where $F$ is either value function. Letting $t\downarrow0$ and using upper semicontinuity proves continuity at $a$. If $\phi_{\max}=0$, the finite domain of $V_I$ is a singleton.

Thus $V_{NI}$ is continuous on $[0,\infty)$ and $V_I$ is continuous on $[0,\phi_{\max}]$, with an upper semicontinuous extension beyond that domain. It follows that $D$ is upper semicontinuous. Therefore
\[
  \Phi_D:=\{\phi\ge0:D(\phi)\ge0\}
\]
is closed. Monotonicity implies that $\phi\in\Phi_D$ entails $[0,\phi]\subseteq\Phi_D$. Since $\Phi_D\subseteq[0,\phi_{\max}]$, it is either empty or a closed interval $[0,\phi_I]$, with $\phi_I\le\phi_{\max}\le\bar\phi$.

\emph{Part 2: FOSD makes the investment region nonempty.} By \eqref{eq:gap-bound}, every feasible contract has $U_k-U_{k-1}\ge0$. Under FOSD, \eqref{eq:abel-wedge} therefore gives
\[
  B(z)=\sum_{k=2}^n(G'_k-G_k)(U_k-U_{k-1})\ge0
  \qquad\text{for every }z\in\mathcal P_n.
\]

Thus the inducing constraint at zero excludes no screening contract, so $V_I(0)=V_{g'}$. The deterrence program adds a constraint to the screening benchmark under $g$, giving $V_{NI}(0)\le V_g$. By Proposition~\ref{prop:Vf-ntype}, $D(0)=V_{g'}-V_{NI}(0) \ge V_{g'}-V_g\ge0$. Hence $\Phi_D$ is nonempty.

Now suppose $D(0)>0$. This requires $g'\ne g$: otherwise both programs at zero coincide. It also requires $\bar C>0$: if $\bar C=0$, reporting IC forces common transfers, and both values at zero equal $\max_{\tau\in\X}V(\tau)$. Choose an experiment $\tau$ with $C(\tau)>0$ and assign it to every type at the common transfer $\theta_1C(\tau)$. This contract satisfies IC and IR, and its investment wedge is
\[
  C(\tau)\sum_{k=2}^n
  (G'_k-G_k)(\theta_{k-1}-\theta_k)>0.
\]

Strictness follows because FOSD and $g'\ne g$ imply that at least one tail difference is positive. Thus $\phi_{\max}>0$. Right-continuity of $D$ at zero now gives $D(\phi)>0$ on some $[0,\epsilon)$ with $\epsilon>0$.

\emph{Part 3: the two-type case.} For $n=2$, set $\theta_1=\thL$, $\theta_2=\thH$, $g=(1-\alpha,\alpha)$, and $g'=(1-\beta,\beta)$. Then $B(z)=(\beta-\alpha)(U_{\thH}-U_{\thL})$, while the objectives are $\beta S_{\thH}+(1-\beta)S_{\thL}$ under investment and $\alpha S_{\thH}+(1-\alpha)S_{\thL}$ under deterrence. Thus the two programs coincide with \eqref{eq:PI} and \eqref{eq:PNI}, respectively, recovering Proposition~\ref{prop:cutoff}.
\end{proof}

\subsection{Proof of Proposition~\ref{prop:overinvestment}}
\label{app:overinvestment}
\begin{proof}
Let the state be binary, with prior $1/2$, posterior $x\in[0,1]$, and $c(x)=\ln2-H(x)$, where $H(x)=-x\ln x-(1-x)\ln(1-x)$, with the usual continuous endpoint convention. Set
\[
\begin{aligned}
    v(x)=\max\{0,3x-1.6,1.4-3x\},\qquad
 \theta=(3/2,1,1/2),\\
 g=(1/4,7/10,1/20),\qquad g'=(1/20,1/20,9/10).
\end{aligned}
\]

The efficient-tail probabilities are $(1,3/4,1/20)$ and $(1,19/20,9/10)$, respectively, so the shift satisfies FOSD. Both distributions have full support.

For any $\lambda>0$, the function $v(x)-\lambda c(x)$ is symmetric around $1/2$. Its pointwise maximum can therefore be attained by a Bayes-plausible experiment placing equal probability on a symmetric pair of maximizers, or by the uninformative experiment. Maximizing over the three action branches gives
\[
  K_\lambda
  =\max\left\{
    0,\,
    \lambda\ln\left(1+e^{3/\lambda}\right)
    -1.6-\lambda\ln2
  \right\}.
\]

Consequently, the first-best investment threshold is
\[
  \phi^{FB}_3
  =-\tfrac15 K_{3/2}
   -\tfrac{13}{20}K_1
   +\tfrac{17}{20}K_{1/2}
  \approx0.295.
\]

Under $g$, the screening virtual multipliers are $(3,29/28,1/2)$, in decreasing order. By Lemma~\ref{lem:ntype-screen}, the corresponding pointwise maximizers, together with minimal rents, attain the screening value $V_g =\tfrac14 K_3+\tfrac7{10}K_{29/28} +\tfrac1{20}K_{1/2} \approx0.634$. Every deterrence contract is feasible for this unconstrained screening problem, so $V_{NI}(\phi)\le V_g$.

We now construct an investment-inducing contract. For $i=1,2$, let $x_i\in(1/2,1)$ solve
\[
  \ln2-H(x_1)=0.02,
  \qquad
  \ln2-H(x_2)=0.50,
\]
and assign $\chi_i=\tfrac12\delta_{x_i}+\tfrac12\delta_{1-x_i}$. Assign full information to type~3. These experiments are Bayes plausible and have costs $C=(0.02,0.50,\ln2)$.

Choose rents and transfers
\[
  U=(0,0.02,0.02+\tfrac12\ln2),
  \qquad T_i=\theta_iC_i+U_i.
\]

The adjacent rent gaps satisfy
\[
\begin{aligned}
  0.01&\le U_2-U_1=0.02\le0.25,\\
  0.25&\le U_3-U_2=\tfrac12\ln2
       \le\tfrac12\ln2.
\end{aligned}
\]

Together with monotone costs, these bounds imply all pairwise reporting constraints. Indeed, summing the adjacent bounds gives
\[
  (\theta_i-\theta_j)C_i
  \le U_j-U_i
  \le(\theta_i-\theta_j)C_j
  \qquad(i<j).
\]

All rents and transfers are nonnegative. The binding adjacent reporting constraint is type~2's constraint against reporting type~3.

The investment incentive is
\[
  \langle g'-g,U\rangle
  =\tfrac15(0.02)
   +\tfrac{17}{20}\left(\tfrac12\ln2\right)
  \approx0.299.
\]

The posterior roots satisfy $x_1\approx0.600$ and $x_2\approx0.952$. The assigned experiments yield $V(\chi_i)=3x_i-1.6$ for $i=1,2$, and $V(\chi_3)=1.4$. Hence the principal's payoff under investment is $\sum_i g'_i[V(\chi_i)-T_i] \approx0.663$.

Conservative bounds on these quantities establish the strict comparisons:
\[
  \phi^{FB}_3<0.296<0.297<0.298
  <\langle g'-g,U\rangle,
\]
and
\[
  V_{NI}(0.297)\le V_g<0.635<0.663
  <\sum_i g'_i[V(\chi_i)-T_i]\le V_I(0.297).
\]

To sum up, at $\phi=0.297$, the change in first-best total surplus net of investment cost is $\phi^{FB}_3-0.297<0$. Nevertheless, the constructed contract strictly induces investment and outperforms every deterrence contract.
\end{proof}

\newpage

\bibliographystyle{plainnat}
\bibliography{references}

\end{document}